\documentclass[10pt,twoside,twocolumn,final]{IEEEtran}

\usepackage{float}
\usepackage{algorithm}
\usepackage[noend]{algorithmic}
\usepackage{bbm}
\usepackage{yhmath}
\usepackage{amscd}
\usepackage{dsfont}
\usepackage{amsmath,amsthm,amssymb,amsfonts}
\usepackage{epsfig}
\usepackage{graphics}
\usepackage{psfrag}
\usepackage{rotating}
\usepackage{url}
\usepackage{color}
\usepackage{balance}
\usepackage{mathtools}
\usepackage{tabularx}
\usepackage{booktabs}
\usepackage{enumitem}
\makeatletter\def\caption@documentclass{standard}\makeatother
\usepackage[labelfont=bf,format=plain,justification=raggedright,singlelinecheck=false]{caption}

\newtheorem{theorem}{Theorem}

\newtheorem{proposition}{Proposition}

\newtheorem{remark}{Remark}

\AtBeginDocument{%
  \captionsetup[figure]{labelfont=bf,font=small,position=bottom,labelsep=colon,justification=justified,singlelinecheck=off}%
}

\begin{document}
\title{Quantum Simultaneous Information and Power Transfer: Capacity-Power Trade-offs in Discrete and Continuous Channels }

 \author{Nizar Khalfet, \IEEEmembership{Member,~IEEE,} and Ioannis Krikidis, \IEEEmembership{Fellow,~IEEE}
\thanks{N. Khalfet and I. Krikidis are with Department of Electrical and Computer Engineering, University of Cyprus, Cyprus (e-mail: \{khalfet.nizar,~krikidis\}@ucy.ac.cy).}
\thanks{This work has received funding from the European Research Council (ERC)
under the European Union’s Horizon 2020 research and innovation programme (Grant agreement No. 819819) and from the
    European Union Marie Sk\l{}odowska-Curie Actions Project
    {COALESCE} under Grant~101130739. This work was supported also by the European Research Council (ERC) through the
    European Union's Horizon Europe Research and Innovation programme,
    under Grant~101241675 ({ERC PoC QUATRO}).}
}
\maketitle

\begin{abstract}
This paper introduces a new framework for quantum simultaneous information and power transfer (QSIPT), enabling the joint use of quantum states for classical information and energy transfer in quantum communication systems. We propose a novel model in which quantum states are simultaneously used to transmit classical information through a quantum channel and transfer energy to an energy harvesting (EH) receiver. The trade-off between communication rate and harvested energy is characterized by the {capacity-power function}, which is defined and characterized for both discrete-variable (DV) and continuous-variable (CV) quantum channels. For DV channels, {we derive the properties of the capacity-power function, providing analytical upper and lower bounds for the amplitude damping channel and an exact closed-form characterization for the quantum erasure channel.} For CV channels, {we extend the mathematical framework by introducing a generalized beam-splitter (BS) receiver with adjustable transmissivity, jointly optimized with a transmitter mean-photon-number budget, that splits the channel output between the information decoder and the EH receiver}. Specifically, we analyze the capacity-power and achievable-rate–power trade-off under various Gaussian encoding schemes including coherent, squeezed, and thermal states for both lossy bosonic and additive Gaussian noise channels. {Closed-form expressions are derived for coherent-state encoding under the joint photon-number-budget and adjustable-transmissivity formulation; squeezed-state inputs are evaluated numerically.} Our results show that, within the considered displaced Gaussian encoding class, coherent states achieve the best capacity-power trade-off, squeezed states do not outperform coherent-state encoding under the phase-insensitive channel and passive receiver architecture, and thermal states enable energy transfer without supporting reliable communication.
\end{abstract}


\section{Introduction}
Quantum communication exploits quantum-mechanical phenomena to enable information processing and communication capabilities that are unattainable using classical technologies. Unlike classical channels, which can only transmit classical information, quantum channels provide a versatile platform for transmitting classical information, quantum information, and entanglement-assisted classical information. Quantum communication exploits principles such as superposition, entanglement, quantum tunneling, etc., to achieve outcomes unattainable by classical systems, such as absolute security and higher transmission rates \cite{b1}. The capacity of a quantum communication channel quantifies the maximum rate at which information can be transmitted reliably between a sender and a receiver. Unlike classical channels, whose capacity is fully characterized by Shannon theory, quantum channels introduce multiple capacity notions due to state non-orthogonality and entanglement \cite{b14,b15}. This foundational understanding of quantum channel capacity provides a theoretical baseline for investigating joint information and energy transfer in quantum systems \cite{QSWIPT}. This novel concept for quantum systems draws inspiration from radiofrequency (RF) signals and the associated technology of simultaneous wireless information and power transfer (SWIPT). In SWIPT, RF signals are ingeniously co-designed to serve a dual purpose i.e., simultaneously conveying both information and energy, through appropriate signal processing and hardware engineering \cite{b2}. The idea of wireless power transfer dates back to the early work of Tesla and has since evolved into a key enabling technology for modern low-power communication systems, sensor networks, machine-type networks, and body-area networks \cite{b4}.  Varshney \cite{b5} first formalized the information-energy capacity region for SWIPT in point-to-point scenarios. Subsequent research has expanded this framework to include parallel links point-to-point channels \cite{b6}, and more recent research has explored the integration of SWIPT into more complex network topologies, such as multiple access channels \cite{b7}, interference channels \cite{b8}, multiple-input multiple-output systems \cite{b9}, and multiple-antenna cellular networks \cite{b10}. A comprehensive overview of existing results in SWIPT for various fundamental multi-user channels is presented in \cite{b11}.

Quantum technologies introduce fundamentally new capabilities in communication systems \cite{Hotta2023_QET}, enabling intrinsic randomness, enhanced security guarantees, and transmission mechanisms that go beyond the limitations of classical binary-based digital architectures \cite{b12,b13}. When leveraged correctly, the computational power of quantum technologies can support the execution of tasks much faster and beyond the capabilities of current digital systems \cite{Preskill2018}.  In particular, recent efforts  explore how quantum signals can be used not only to transmit classical or quantum information but also to deliver energy to quantum devices, thereby opening new directions in energy-aware quantum communication design \cite{Meng2025_quantumHarvester}. 

The integration of energy harvesting (EH) mechanisms within quantum information-theoretic models  \cite{Khalfet2025_QSIPT} opens a novel frontier for communication system design. While classical SWIPT uses the dual role of RF signals to enable simultaneous information and power transfer, its extension to the quantum systems introduces fundamentally new challenges and opportunities. In our work, we propose the framework of quantum simultaneous information and power transfer (QSIPT), where quantum signals convey classical information and deliver usable energy to an EH receiver. This paradigm is motivated by the growing need for energy-efficient operation in quantum devices such as quantum internet of things, quantum sensors, embedded receivers, and networked quantum nodes. The unique structure of quantum states and the nature of quantum noise significantly alter the trade-offs compared to classical SWIPT, requiring a new information-theoretic framework. With the exception of a few studies (e.g., \cite{QSWIPT,Khalfet2025_QSIPT}), existing works do not consider the co-design of information and energy in quantum communication from an information theory standpoint. In addition, most of the aforementioned studies focus on the discrete variable (DV) case and therefore the impact of the EH constraints on  continuous variable (CV) channels has not been investigated. 

Although CV quantum channels constitute the dominant scenario for current optical implementations, this work presents both DV and CV analyses within a unified QSIPT framework. The DV setting plays a foundational role by enabling a transparent and, in several cases, analytical characterization of the capacity-power function in finite-dimensional systems. In particular, the DV analysis reveals fundamental structural properties of the capacity-power function, including monotonicity, the emergence of capacity plateaus, and transitions in the optimal signaling strategy as the EH requirement increases. These behaviors provide valuable conceptual benchmarks that guide the interpretation of the  CV results, where closed-form characterizations are generally unavailable and numerical optimization is required. Moreover, the DV study highlights mechanisms such as boundary activation of the EH constraint and the effective loss of signaling degrees of freedom at high EH levels, which also manifest in CV channels. Finally, beyond optical CV systems, DV channels are directly relevant to emerging finite-dimensional quantum technologies, such as superconducting and solid state quantum devices, where energy extraction and dissipation constraints are intrinsic.

This paper establishes a unified mathematical framework of QSIPT by characterizing the fundamental trade-offs between classical communication and EH  in quantum channels. A novel capacity-power function $C_Q(B)$ is introduced to quantify the trade-off for both  DV and  CV quantum channels.  For the DV case, we examine the amplitude damping channel and the quantum erasure channel, which serve as guideline  models for the dissipative and lossy information transfer, respectively. For the CV case, we investigate the lossy bosonic channel and the additive Gaussian noise (AGN) channel, which represent two fundamental models in optical quantum communication. In each scenario, we derive the capacity-power functions that quantify the trade-off between information transmission and harvested energy, under physically consistent constraints and modulation schemes.
Our contributions are twofold:
\begin{itemize}
\item \textbf{Capacity-Power Analysis for DV quantum Channels}: We study the amplitude damping and quantum erasure channels as representative DV models, offering theoretical insights for the QSIPT framework. {For the amplitude damping channel, we derive analytical upper and lower bounds on the capacity-power function $C_Q(B)$, revealing a transition from maximally mixed to energy-biased input ensembles as the harvested energy constraint increases. In the case of the quantum erasure channel, we obtain a closed-form expression for $C_Q(B)$ under the standard qubit normalization $E_0=0$, $E_1=1$, which is shown to be strictly concave and monotonically decreasing in the active EH regime.} Although CV channels dominate current optical implementations, the DV setting plays a foundational analytical role for the more practical CV scenarios.
\item \textbf{Capacity-Power Analysis for CV Channels}: {We propose a unified QSIPT framework for CV quantum channels by modeling the co-located receiver via a beam splitter with an adjustable transmissivity that controls the fraction of received optical power directed to the information decoder, and by imposing a mean-photon-number budget at the transmitter. The capacity-power function is characterized through the joint optimization of the input photon number and the receiver transmissivity subject to the EH constraint. For the lossy bosonic and AGN channels, we obtain a closed-form characterization of the capacity-power function under coherent-state encoding, in which the optimal strategy operates the transmitter at its full photon budget and realizes the trade-off entirely through the receiver beam-splitter architecture. Across all tested parameter configurations, numerical optimization over the considered displaced Gaussian encoding family (including squeezed states) yields zero squeezing as the optimum, indicating that coherent-state encoding achieves the best performance among the considered Gaussian inputs. This finding is consistent with the phase-insensitive nature of the channel and the passive beam-splitter receiver architecture, which do not enable the exploitation of quadrature asymmetry introduced by squeezing, and with the known optimality of coherent states for phase-insensitive Gaussian channels under energy constraints~\cite{Giovani,Holevo2001}.} Compared to DV models, CV channels have greater implementation interest, making them more suitable for practical QSIPT designs.
\end{itemize}
\section{Background and Notation}
In this section, we introduce the key concepts and notation for  DV and  CV quantum communication systems. These definitions provide the foundational understanding required for the analysis presented in the remainder of the paper.
\subsection{Discrete-variable quantum communication}
In  DV quantum systems, information is encoded into quantum states acting on a finite-dimensional Hilbert space $\mathcal{H}_d$. A classical input symbol $x$ from a finite alphabet $\mathcal{X}$ is mapped to a quantum state $\rho_x$, where $\rho_x$ is a density operator on $\mathcal{H}_d$. 
We denote by $\{p_x, \rho_x\}$ the ensemble of input states, where $p_x$ is the probability of transmitting the quantum state $\rho_x$.  The average state of the ensemble is given by
\begin{equation}
	\label{eq:avg_state}
    \rho = \sum_{x \in \mathcal{X}} p_x \rho_x.
\end{equation}
The von Neumann entropy of a quantum state $\rho$ is given by
\begin{equation}
\label{eq:vn_entropy}
S(\rho) = -\mathrm{Tr}[\rho \log_2 \rho],
\end{equation}
where $\mathrm{Tr}[\cdot]$ denotes the \emph{trace} of an operator, defined as the sum of its diagonal elements (or equivalently, the sum of its eigenvalues).
We denote by $|0\rangle$ and $|1\rangle$ the computational basis vectors of a two-level quantum system (qubit). A pure state $|\psi\rangle$ in this system can be written as $|\psi\rangle = \alpha |0\rangle + \beta |1\rangle$, where $\alpha, \beta \in \mathbb{C}$ and $|\alpha|^2 + |\beta|^2 = 1$.
The Hamiltonian operator $H$ is a Hermitian operator that models the energy levels of the quantum system. In the case of a single qubit, the Hamiltonian typically takes the diagonal form
\begin{equation}
\label{eq:hamiltonian}
    H = 
    \begin{pmatrix}
        E_0 & 0 \\
        0 & E_1
    \end{pmatrix},
\end{equation}
where $E_0$ and $E_1$ denote the energy of the ground and excited state, respectively. {Throughout this work we adopt the standard qubit normalization $E_0=0$, $E_1=1$, which sets the unit of energy to the level spacing of the qubit and ensures that $\mathrm{Tr}[H\rho]\in[0,1]$ for any qubit state $\rho$. This convention is used in all subsequent DV analyses, including Theorems~1 and~2 and the corresponding appendices.} We denote by $ \langle \cdot \rangle \equiv \mathrm{Tr}[H \rho]$ the expected energy of the quantum state $\rho$ under the Hamiltonian $H$. This quantity plays a central role in modeling  EH mechanisms in DV quantum communication systems.
\subsection{Continuous-variable quantum communication}
In CV quantum systems \cite{Usenko2025_CVReview}, information is encoded into quadrature observables of bosonic modes, which are modeled by infinite-dimensional Hilbert spaces. These systems are particularly relevant in quantum optics, where light fields serve as carriers of quantum information. We denote by $\hat{a}$ and $\hat{a}^\dagger$ the bosonic annihilation and creation operators, which satisfy the canonical commutation relation 
\begin{equation}
\label{eq:commutator}	
[\hat{a}, \hat{a}^\dagger] = \hat{a} \hat{a}^\dagger - \hat{a}^\dagger \hat{a} = \mathbb{I},
\end{equation}
 where $\mathbb{I}$ denotes the identity operator on the Hilbert space of a single bosonic mode.  This is known as the \emph{commutator} of $\hat{a}$ and $\hat{a}^\dagger$, and it reflects the fundamental non-commutativity of quantum operators. The quadrature operators $\hat{X}$ and $\hat{P}$ are given by
\begin{equation}
	\label{eq:quadratures}
    \hat{X} = \frac{1}{\sqrt{2}}(\hat{a} + \hat{a}^\dagger), \quad 
    \hat{P} = \frac{1}{\sqrt{2}i}(\hat{a} - \hat{a}^\dagger),
\end{equation}
which correspond to the position and momentum quadratures, respectively, and satisfy $[\hat{X}, \hat{P}] = i$.

{A coherent state $|\alpha\rangle$ is the eigenstate of the annihilation operator $\hat{a}$, satisfying}
\begin{equation}
    \hat{a} |\alpha\rangle = \alpha |\alpha\rangle, \quad \alpha \in \mathbb{C}.
\end{equation}
Coherent states are generated by applying the displacement operator $D(\alpha) = \exp(\alpha \hat{a}^\dagger - \alpha^* \hat{a})$ to the vacuum state, yielding $|\alpha\rangle = D(\alpha) |0\rangle$. These states exhibit symmetric quadrature variances,
\begin{equation}
\label{eqVarCoher}
    \mathrm{Var}(\hat{X}) = \mathrm{Var}(\hat{P}) = \frac{1}{2}.
\end{equation}
{These states carry a mean photon number $N_s = |\alpha|^2$.} {A squeezed vacuum state $|r\rangle$ is obtained by applying the squeezing operator as follows
}
\begin{equation}
\label{eqSq}
S(r) = \exp\left[\frac{r}{2}(\hat{a}^2 - \hat{a}^{\dagger 2})\right],
\end{equation}
to the vacuum, where $r \in \mathbb{R}$ is the squeezing parameter, i.e., 
\begin{equation}
\label{eqSqrho}
|r\rangle = S(r) |0\rangle,
\end{equation}
The resulting state exhibits asymmetric quadrature variances,
\begin{equation}
\label{eqVarSq}
    \mathrm{Var}(\hat{X}) = \frac{1}{2}e^{-2r}, \quad 
    \mathrm{Var}(\hat{P}) = \frac{1}{2}e^{2r},
\end{equation}
and its mean photon number is given by $N_s = \sinh^2 (r)$.
{A thermal state $\rho_{\text{th}}$ is a mixed state representing a bosonic mode in thermal equilibrium. It is diagonal in the Fock basis, i.e.,
}
\begin{equation}
\label{eqRhoTh}
    \rho_{\text{th}} = \sum_{n=0}^\infty \frac{N_{s}^n}{(1 + N_{s})^{n+1}} |n\rangle \langle n|,
\end{equation}
where {$N_s$ denotes the mean photon number of the thermal state.} Thermal states have equal quadrature variances, i.e., 
\begin{equation}
    \mathrm{Var}(\hat{X}) = \mathrm{Var}(\hat{P}) = \frac{1}{2}(1 + 2N_{s}),
\end{equation}
and are used to model background noise and classical signal mixtures. We denote by $\hat{n} = \hat{a}^\dagger \hat{a}$ the photon number operator, and by $\mathrm{Tr}[\hat{n} \rho]$ the mean photon number of a state $\rho$, which corresponds to the expected energy in natural units.  These coherent, squeezed, and thermal states form the basis of practical CV quantum communication schemes.

\subsection{QSIPT topology}
\begin{figure}\centering
\includegraphics[width=\linewidth]{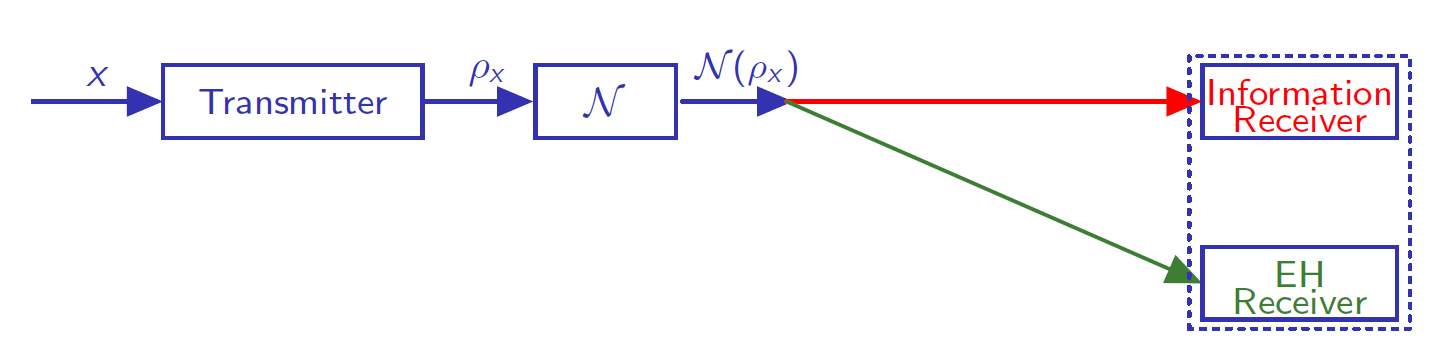}
\caption{A co-located QSIPT system over a DV quantum channel.} \label{model}
\end{figure}
We consider a three-party quantum communication setup with a co-located information and EH receiver, as illustrated in Fig.~\ref{model}, where a single receiver simultaneously performs information decoding and EH on the channel output. The transmitter encodes classical symbols $x$ into quantum states $\rho_x$ with prior probabilities $p_x$, which are then transmitted through a  quantum channel $\mathcal{N}$. Specifically, we study two representative cases within this general topology: DV quantum channels and CV quantum channels. For the DV case, the average EH is expressed through the expectation $\mathrm{Tr}[H \,\mathcal{N}(\rho)]$, whereas, {in the CV case, co-location is modeled by a beam splitter with adjustable transmissivity $\tau\in[0,1]$ at the receiver, which directs a fraction $\tau$ of the received optical power to the information decoder and the complementary fraction $(1-\tau)$ to the EH receiver. The transmissivity $\tau$ is a design parameter jointly optimized with the input mean photon number $N_s\in[0,N_{\max}]$, where $N_{\max}$ is the transmitter photon-number budget.}

The remainder of the paper is organized as follows. Section~III analyzes the capacity-power trade-off in DV quantum channels, focusing on the amplitude damping and quantum erasure models. Section~IV  extends the analysis to CV quantum channels under the BS based co-located architecture, considering coherent, squeezed, and thermal state encodings. Section~V presents numerical simulations, while {Section~VI} concludes the paper.
\section{Capacity-power Analysis for DV Channels}
 For the DV case, the receiver applies a positive operator-valued measure (POVM) to decode the information, and the accessible information is bounded by the Holevo quantity \( \chi \), defined as
\begin{equation}
\chi(\xi) = S(\rho) - \sum_x p_x S(\rho_x),
\end{equation}
where \( S(\cdot) \) is defined in \eqref{eq:vn_entropy} and \( \xi = \{p_x, \rho_x\} \) is the input ensemble. {When the ensemble $\xi$ is transmitted through a channel $\mathcal{N}$, the Holevo quantity of the resulting output ensemble is denoted $\chi(\xi,\mathcal{N}):=\chi(\mathcal{N}(\xi))$, where $\mathcal{N}(\xi)=\{p_x,\mathcal{N}(\rho_x)\}$ is the output ensemble induced by $\xi$ through $\mathcal{N}$.}

For a DV quantum channel \( \mathcal{N} \), the classical capacity \( C_Q(\mathcal{N}) \) is defined as the maximum Holevo information over all possible input ensembles
\begin{equation}
C_Q(\mathcal{N}) = \max_{\xi} \chi(\xi,\mathcal{N}).
\end{equation}
where \( \mathcal{N}(\xi) \) refers to the action of the quantum channel \( \mathcal{N} \) on the input ensemble $\xi$, mapping it to the output ensemble \( \mathcal{N}(\xi) = \{p_x, \mathcal{N}(\rho_x)\} \). The ensemble-average input state is denoted by
\begin{equation}
\rho = \sum_x p_x \rho_x,
\end{equation}
and the corresponding average output state is given by $\mathcal{N}(\rho)$.
{In the DV case, the harvested energy per channel use is defined as the expectation value of the output Hamiltonian, $\mathrm{Tr}[ H \mathcal{N}(\rho) ]$. This quantity represents the average energy carried by the received quantum state during a single signaling interval. Over multiple channel uses, the total harvested energy accumulates additively, consistent with the assumption of independent channel uses and an additive Hamiltonian.} Therefore, the capacity-power function \( C_Q(B) \) is then defined as the maximum transmission rate achievable under the constraint that the harvested output energy is at least \( B \), i.e., 
{\begin{equation}
        C_Q(B) =
\max_{\substack{\xi=\{p_x,\rho_x\} \\
\mathrm{Tr}[H\,\mathcal{N}({\rho})]\ge B}}
\chi\!\left(\xi,\mathcal{N}\right),
     \end{equation}
}
{where $\chi\!\left(\xi,\mathcal{N}\right)$ denotes  the Holevo quantity of the output ensemble induced by $\xi$ through the channel $\mathcal{N}$.  }
For a blocklength \( n \), this definition generalizes to
\begin{IEEEeqnarray}{l}
\nonumber
C_Q^{(n)}(B) = \max_{\{p_{X^n}, \rho_{X^n}\} : {\mathrm{Tr}}[H^{\otimes n} \mathcal{N}^{\otimes n}(\rho^{ (n)})] \geq nB} \frac{1}{n} \bigg[ S(\mathcal{N}^{\otimes n}(\rho^{ (n)}))\\
- \sum_{x^n} p_{x^n} S(\mathcal{N}^{\otimes n}(\rho_{x^n})) \bigg],
\end{IEEEeqnarray}
where the notation \( H^{\otimes n} \) refers to the \( n \)-fold tensor product of the Hamiltonian operator \( H \), defined as
\begin{equation}
H^{\otimes n} = H \otimes H \otimes \cdots \otimes H \quad \text{(n times)},
\end{equation}
where \( H \) models the energy levels of a single quantum system.
{The use of the additive Hamiltonian $H^{\otimes n}$ reflects the assumption that EH is performed independently on each channel use. Specifically, this model captures the behavior of practical EH receivers, such as photodetection or rectenna-based architectures, which operate locally on individual modes and do not implement coherent multi-mode measurements. While more general $H_n$, in principle, exploit quantum correlations to enhance EH, such measurements would require coherent inter-mode interactions and are beyond the scope of the present framework. Accordingly, we restrict our analysis to additive energy observables and leave the study of collective EH strategies for future work \cite{Campaioli_2017}.}

\begin{remark}
    {Throughout this work, we consider the one-shot Holevo capacity, corresponding to product-state encoding without entanglement across channel uses. While regularization may enhance capacity for certain channels (e.g., amplitude damping), the trade-off characterized by $C_Q(B)$ remains operationally meaningful for practical systems where coding over long blocks may be infeasible under  EH constraints \cite{Shirokov_2020}.}
\end{remark}
\subsection{Properties of the capacity-power function}
In this subsection, we analyze the structural properties of the capacity-power function $C_Q(B)$ for DV quantum channels under EH constraints. We consider channels $\mathcal{N}$ with finite-dimensional input and output spaces, and we assume the energy observable is given by a Hermitian Hamiltonian operator $H$ acting on the channel output. We now state and prove several  properties of $C_Q(B)$.
\begin{proposition}[Monotonicity]
 The function $C_Q(B)$ is non-increasing in $B$.
\end{proposition}
\begin{IEEEproof}
Let $B_1 < B_2$ and let $\mathcal{F}_B$ denote the feasible set of input ensembles satisfying $\mathrm{Tr}[H \mathcal{N}(\rho)] \geq B$. Then $\mathcal{F}_{B_2} \subseteq \mathcal{F}_{B_1}$, which implies
\[
C_Q(B_2) = \max_{\xi \in \mathcal{F}_{B_2}} \chi(\xi,\mathcal{N}) \leq \max_{\xi \in \mathcal{F}_{B_1}} \chi(\xi,\mathcal{N}) = C_Q(B_1).
\]
\end{IEEEproof}
\begin{proposition}[Piecewise Concavity]
   The function $C_Q(B)$ is piecewise concave. 
\end{proposition}
\begin{IEEEproof}
 The Holevo quantity is a concave function of the average output state $\mathcal{N}(\rho)$ and a convex function of the ensemble. However, the constraint $\mathrm{Tr}[H \mathcal{N}(\rho)] \geq B$ defines a convex feasible set in terms of the ensemble average state. The optimization over both probabilities $p_x$ and states $\rho_x$ results in a supremum of concave functions over a piecewise-defined domain (due to potential changes in the optimal number of states or structure), hence the resulting function $C_Q(B)$ is piecewise concave.
\end{IEEEproof}
\begin{proposition}[Continuity]
The function $C_Q(B)$ is upper semi-continuous and left-continuous.
\end{proposition}
\begin{IEEEproof}
 The set of admissible input ensembles is compact under the trace-norm topology (due to normalization and finite dimension), and the Holevo quantity is continuous with respect to the ensemble. The supremum over a compact set of continuous functions is upper semi-continuous. Since the feasible set continuously decreases with increasing $B$, the left-continuity also holds. Right-continuity may fail at thresholds where the feasible set becomes empty.
\end{IEEEproof}
\begin{proposition}[Capacity Plateau at Low Energy]
 If $B \leq B_0$, where $B_0 = \mathrm{Tr}[H \mathcal{N}(\rho_{\star})]$ for a capacity-achieving ensemble $\rho_{\star}$ without energy constraints, then
\[
C_Q(B) = C_Q(0), \quad \forall B \leq B_0.
\]
\end{proposition}
\begin{IEEEproof}
 For any $B \leq B_0$, the unconstrained optimal ensemble $\rho_{\star}$ remains feasible since $\mathrm{Tr}[H \mathcal{N}(\rho_{\star})] \geq B$. Hence, the maximum value of the Holevo quantity is constant. {Therefore, $C_Q(B) = \chi(\xi_{\star}, \mathcal{N}) = C_Q(0)$,
where $\xi_{\star} = \{p_x^{\star}, \rho_x^{\star}\}$ denotes the capacity-achieving input ensemble.}
\end{IEEEproof}
{
\begin{remark}[Bias induced by the EH constraint in DV channels]
For a binary input ensemble $\xi=\{(1-q,\rho_0),(q,\rho_1)\}$ with $q\in[0,1]$, the ensemble-average input state is $\bar{\rho}(q)=(1-q)\rho_0+q\rho_1$. By linearity of the channel and of the trace operator, the harvested energy at the channel output can be written as
\begin{equation}
B(q)=\mathrm{Tr}\!\big[H\,\mathcal{N}(\bar{\rho}(q))\big]=(1-q)B_0+qB_1,
\end{equation}
where $B_i\triangleq \mathrm{Tr}[H\,\mathcal{N}(\rho_i)]$, $i\in\{0,1\}$. Hence, the EH constraint $B(q)\geq B$ induces an explicit lower bound on the bias parameter, namely $q \geq (B-B_0)/(B_1-B_0)$, whenever $B_1>B_0$. This shows that increasing energy requirements progressively bias the feasible input ensembles toward energy-carrying states, thereby restricting the admissible probability simplex and reducing the achievable Holevo information.
\end{remark}
}

 To illustrate the impact of the EH constraint on the quantum channel performance, we will focus on two indicative examples, i.e., the amplitude damping channel, and the erasure channel. The amplitude damping channel models energy dissipation processes, such as spontaneous emission, while the erasure channel captures the probabilistic loss of information. 
\subsection{Amplitude damping channel}
{The amplitude damping channel models energy dissipation mechanisms in quantum systems, such as spontaneous emission, and is particularly relevant for physical qubit implementations based on energy-level transitions, including atomic and superconducting platforms.}

The Kraus operators for the amplitude damping channel are defined as
{
 \noindent \begin{equation}
    K_0=\begin{pmatrix}1&0\\0&\sqrt{1-\gamma}\end{pmatrix}, \quad K_1=\begin{pmatrix}0&\sqrt{\gamma}\\0&0\end{pmatrix} .
     \end{equation}
     }
where \( \gamma \) is the damping parameter, representing the probability of energy loss. As \( \gamma \) increases, the probability of transitioning from the excited state to the ground state increases, leading to energy dissipation.  For an input quantum state \( \rho_{\text{in}} \), the output state \( \mathcal{N}(\rho) \) after passing through the amplitude damping channel is given by
\begin{equation}
\mathcal{N}(\rho) = K_0 \rho K_0^\dagger + K_1 \rho K_1^\dagger.
\end{equation}
Substituting the Kraus operators, the output state becomes
\begin{equation}
\mathcal{N}(\rho) = \begin{pmatrix} 
\rho_{00} + \gamma \rho_{11} & \sqrt{1 - \gamma} \rho_{01} \\
\sqrt{1 - \gamma} \rho_{10} & (1 - \gamma) \rho_{11} 
\end{pmatrix},
\end{equation}
where \( \rho_{ij} \) are the elements of the input density matrix \( \rho_{\text{in}} \). Using the Hamiltonian  in \eqref{eq:hamiltonian}, the expected energy of the output state is given by
\begin{equation}
{\mathrm{Tr}}[H \mathcal{N}(\rho)] = E_0 \mathcal{N}(\rho)_{00} + E_1 \mathcal{N}(\rho)_{11}
\end{equation}
Substituting the elements of \( \mathcal{N}(\rho) \), we have
\begin{equation}
{\mathrm{Tr}}[H \mathcal{N}(\rho)] = E_0 (\rho_{00} + \gamma \rho_{11}) + E_1 (1 - \gamma) \rho_{11}.
\end{equation}
Thus, the EH constraint is expressed as
\begin{equation}
E_0 (\rho_{00} + \gamma \rho_{11}) + E_1 (1 - \gamma) \rho_{11} \geq B,
\end{equation}
where \( B \) is the minimum harvested energy. The optimization problem for the amplitude damping channel under the EH constraint can now be written as

\begin{equation}
\label{EqOptAmp}
\begin{aligned}
\max_{\{p_x, \rho_x\}} \quad & S\left( \sum_x p_x \mathcal{N}(\rho_x) \right) - \sum_x p_x S\left( \mathcal{N}(\rho_x) \right) \\
\text{subject to} \quad & {\mathrm{Tr}}[H \mathcal{N}(\rho)] \geq B, \\
& \rho = \sum_x p_x \rho_x, \\
& {\mathrm{Tr}}[\rho_x] = 1 \quad \forall x, \\
& p_x \geq 0, \quad \sum_x p_x = 1, \\
& \rho_x \succeq 0 \quad \forall x.
\end{aligned}
\end{equation}
To evaluate the capacity-power function for each quantum channel under consideration, we numerically solve the optimization problem defined in \eqref{EqOptAmp} for the respective channel. The objective function, representing the Holevo information or the capacity as a function of energy, is maximized subject to the EH constraints and other physical constraints, such as positivity and normalization of the density matrices. The optimization problems are formulated in standard convex form and solved by using the CVX toolbox, a MATLAB-based numerical solver designed for convex optimization problems \cite{b16}. Based on the analytical expressions derived above, we now state upper and lower bounds for the capacity-power function of the amplitude damping channel. These bounds are summarized in the following theorem
\begin{theorem}
{Let $C_Q(B)$ denote the capacity-power function of the amplitude damping channel with damping parameter $\gamma\in[0,1]$ and Hamiltonian $H=\mathrm{diag}(0,1)$. Then, for all $B\in[0,1-\gamma]$, the following analytical upper and lower bounds hold:}
\begin{align}
C_Q(B) &\leq \max_{0 \leq a \leq 1 - \frac{B}{1 - \gamma}} H_2((1 - \gamma)(1 - a)), \label{eq:upper_bound_AD}\\
C_Q(B) &\geq H_2\left( \frac{1 - \gamma}{2} \right) - \frac{1}{2} H_2(1 - \gamma), \quad \text{for } B \leq \frac{1 - \gamma}{2}, \label{eq:lower_bound_AD}
\end{align}
where $H_2(x) = -x \log_2 x - (1 - x) \log_2 (1 - x)$ is the binary entropy function{, and $a\in[0,1]$ denotes the ground-state population of the diagonal input state $\rho=\mathrm{diag}(a,1-a)$}.
\end{theorem}
\begin{IEEEproof}
The proof is provided in Appendix~\ref{appendix:ad_bounds}.   
\end{IEEEproof}
{The bounds in~\eqref{eq:upper_bound_AD}-\eqref{eq:lower_bound_AD} are derived through distinct analytical constructions: the upper bound originates from a converse argument that maximizes the output von Neumann entropy over diagonal inputs satisfying the EH constraint, while the lower bound is established by an explicit binary pure-state ensemble whose Holevo information is computed in closed form. The two bounds do not coincide in general, as the binary ensemble does not attain the converse upper bound; the gap is given by the strictly positive term $\tfrac12 H_2(1-\gamma)$ for $\gamma\in(0,1)$, reflecting the fact that the upper bound of~\eqref{eq:upper_bound_AD} disregards the contribution of the conditional output entropies to the Holevo quantity.}

\subsection{Quantum erasure channel}
{The quantum erasure channel models a communication scenario in which a qubit is either transmitted correctly or erased with a known probability $p_e$. In the event of an erasure, the receiver is explicitly informed via an erasure flag, which distinguishes this channel from noise models in which errors remain undetected.}
The Kraus operators for the quantum erasure channel are given by
\begin{equation}
K_0 = \sqrt{1 - p_e} \cdot I, \quad K_1 = \sqrt{p_e} \cdot |e\rangle \langle 0|, \quad K_2 = \sqrt{p_e} \cdot |e\rangle \langle 1|,
\end{equation}
where \( p_e \) is the probability of erasure and \( |e\rangle \) is an auxiliary state indicating that the qubit has been erased.
\begin{proposition}
The capacity of the quantum erasure channel without considering the EH is given by
\begin{equation}
C = (1 - p_e) \log_2 d,
\end{equation}
where \( d \) is the dimension of the input space. For a quantum system with a single qubit, \( d = 2 \),  the capacity simplifies to \( C = 1 - p_e \).
\end{proposition}
\begin{IEEEproof}
The proof is presented in Appendix \ref{ProofOfProp}.
\end{IEEEproof}
 The quantum erasure channel models the probabilistic loss of information, where the transmitted state \( \rho_{\text{in}} \) is replaced with a fixed erasure state \( |e\rangle \langle e| \) with probability \( p_e \). The output state \( \mathcal{N}(\rho) \) for this channel is given by
\begin{equation}
\mathcal{N}(\rho) = (1 - p_e)\,{\rho_{\mathrm{in}}} + p_e |e\rangle \langle e|,
\end{equation}
where {$\rho_{\mathrm{in}}$ denotes the input state transmitted without erasure (i.e., received intact with probability $1-p_e$)}.

For a Hilbert space including the states \( |0\rangle \), \( |1\rangle \), and the erasure state \( |e\rangle \), the Hamiltonian is given by extending \eqref{eq:hamiltonian},
where \( E_0 \), \( E_1 \), and \( E_e \) are the energy levels of \( |0\rangle \), \( |1\rangle \), and \( |e\rangle \), respectively. {Throughout this subsection we adopt the standard qubit normalization $E_0=0$, $E_1=1$ introduced in Section~II-A, while $E_e\ge 0$ is treated as a free parameter modeling the energy of the erasure flag.} We now derive the analytical expression for the capacity-power function of the quantum erasure channel under an EH constraint. The result is summarized in the following theorem.
\begin{theorem}
{Under the standard qubit energy normalization $E_0=0$, $E_1=1$, let $E_e\ge 0$ denote the erasure-flag energy. The capacity-power function of the quantum erasure channel with erasure probability $p_e\in[0,1]$ is defined for $B\in[0,(1-p_e)+p_eE_e]$ and given by}
\begin{equation}
C_Q(B) =
\begin{cases}
1 - p_e, & B \leq {\dfrac{1-p_e}{2}+p_eE_e}, \\[5pt]
(1 - p_e) H_2\!\left({\dfrac{B-p_eE_e}{1-p_e}}\right), & B > {\dfrac{1-p_e}{2}+p_eE_e}.
\end{cases}
\label{eq:erasure_capacity_expression}
\end{equation}
\end{theorem}
\begin{IEEEproof}
    The proof is provided in Appendix~\ref{appendix:erasure}.
\end{IEEEproof}

\begin{remark}[Comparison with the amplitude damping channel]
The quantum erasure channel admits a closed-form characterization of the capacity-power function $C_Q(B)$ due to its  simplicity. Specifically, the channel acts as a probabilistic mixture of an  input state and a known orthogonal erasure flag. In contrast, the amplitude damping channel induces  a more complex optimization problem and  requires numerical methods or bounding techniques. Moreover, in the erasure channel, the optimal input remains diagonal and the maximizing ensemble is explicitly known, whereas in the amplitude damping case, the optimal ensemble may involve non-orthogonal states that balance information and energy transfer.
\end{remark}
\begin{remark}[{Strict concavity in the active EH regime}]
{In the active EH regime $B>(1-p_e)/2+p_eE_e$, the optimal input bias is
$a^\star(B) = 1 - q_{\min}(B)$, where
\begin{equation*}
q_{\min}(B) \triangleq \frac{B - p_e E_e}{1 - p_e}
\end{equation*}
is an affine, strictly increasing function of $B$, ranging over $[1/2,1]$ throughout
this regime. Substituting into~\eqref{eq:erasure_capacity_expression} and
using the symmetry $H_2(x)=H_2(1-x)$ yields
\begin{equation*}
C_Q(B) = (1-p_e)\,H_2\!\left(q_{\min}(B)\right).
\end{equation*}
Since $H_2(\cdot)$ is strictly concave on $[0,1]$ and strictly decreasing on
$[1/2,1]$, the composition of $H_2$ with the affine map $B\mapsto q_{\min}(B)$
is \emph{strictly concave and monotonically decreasing} in $B$. In particular,
the affinity of $q_{\min}(B)$ in $B$ does not imply linearity of $C_Q(B)$.}
\end{remark}

\section{Capacity-power analysis for CV channels}
{For the CV case, we extend the analysis to quantum channels defined on infinite-dimensional Hilbert spaces, which naturally arise in optical quantum communication systems. Such channels employ CV quantum states, including coherent and squeezed states of light, and are subject to noise and loss mechanisms intrinsic to these physical channels.} To characterize the capacity-power function,  we consider two prominent case studies i.e., the lossy bosonic channel and the AGN channel. These channels are particularly relevant for scenarios involving coherent and squeezed states of light, where noise and loss mechanisms significantly influence the achievable information and energy rates.
\subsection{General framework for CV Channels}
We consider single-mode CV quantum communication systems where information is encoded into quantum states of an electromagnetic field mode. These systems are modeled over infinite-dimensional Hilbert spaces and described in terms of quadrature operators or annihilation/creation operators. { We consider two representative CV quantum channel models that differ fundamentally in the nature of the noise they introduce. The first is the lossy bosonic channel, in which noise enters multiplicatively through attenuation of the input field and coupling to an environmental mode. This channel is characterized by a transmissivity parameter $T \in [0,1]$ and models energy loss during propagation. The second is the  AGN  channel, in which noise acts additively and in a phase-insensitive manner by directly perturbing the field quadratures. Specifically, let $\hat a_{\text{in}}$ and $\hat a_{\text{out}}$ denote the bosonic annihilation operators of the input and output modes, respectively. Moreover, let $\hat x=(\hat a+\hat a^\dagger)/\sqrt{2}$ and $\hat p=(\hat a-\hat a^\dagger)/(i\sqrt{2})$ denote the canonical field quadrature operators.
}
{
 \begin{itemize}
\item \textbf{Lossy bosonic channel}: multiplicative attenuation, modeled as
\[
\hat{a}_{\text{out}}=\sqrt{T}\,\hat{a}_{\text{in}}+\sqrt{1-T}\,\hat{b},
\]
where $T\in[0,1]$ denotes the channel transmissivity and $\hat b$ is an environmental bosonic mode, typically assumed to be in a thermal state.
\item   \textbf{AGN channel}: additive phase-insensitive noise, modeled in the quadrature domain as
\[
\hat{x}_{\text{out}}=\hat{x}_{\text{in}}+\hat{x}_{\text{noise}},\qquad
\hat{p}_{\text{out}}=\hat{p}_{\text{in}}+\hat{p}_{\text{noise}},
\]
\end{itemize}
}
{where $\hat{x}_{\text{noise}}$ and $\hat{p}_{\text{noise}}$ are independent, zero-mean Gaussian noise operators with equal variance, capturing phase-insensitive additive noise.}

\begin{figure}[t]
    \centering
    \includegraphics[width=1\linewidth]{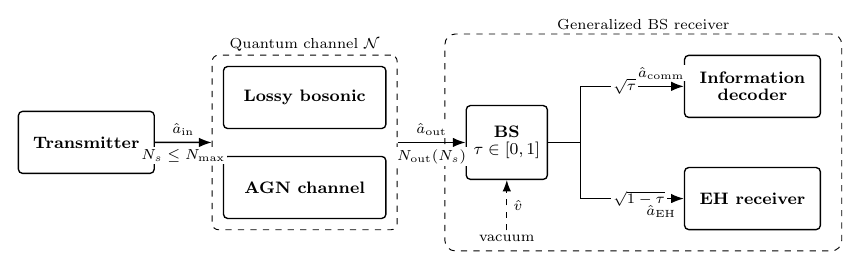}
    \caption{{Generalized BS receiver architecture for CV-QSIPT systems with adjustable transmissivity $\tau\in[0,1]$.}}
    \label{fig:beamsplit}
\end{figure}
{To model simultaneous information decoding and EH in the CV case, we consider a generalized BS receiver with adjustable transmissivity $\tau\in[0,1]$, which directs a fraction $\tau$ of the channel output to the information decoder and the complementary fraction $(1-\tau)$ to the EH receiver, as illustrated in Fig.~\ref{fig:beamsplit} \cite{b18}. Let $\hat{a}_{\mathrm{out}}$ denote the annihilation operator of the channel output and $\hat{v}$ a vacuum mode entering the unused port of the BS. The comm-port and EH-port operators are then given by
\begin{equation}
\hat{a}_{\mathrm{comm}} = \sqrt{\tau}\,\hat{a}_{\mathrm{out}}+\sqrt{1-\tau}\,\hat{v},\quad
\hat{a}_{\mathrm{EH}} = \sqrt{1-\tau}\,\hat{a}_{\mathrm{out}}-\sqrt{\tau}\,\hat{v},
\label{eq:cv_BS_modes}
\end{equation}
so that $\hat{a}_{\mathrm{comm}}$ is used for Holevo decoding and $\hat{a}_{\mathrm{EH}}$ for EH. The symmetric choice $\tau=\tfrac12$ recovers the passive 50:50 BS configuration employed in co-located CV receiver architectures studied in the quantum simultaneous information and power transfer literature~\cite{QSWIPT}. Unlike the symmetric case, the adjustable transmissivity $\tau$ enters the optimization as a design parameter that is jointly chosen with the input statistics, enabling a continuous and operationally meaningful trade-off between the achievable information rate and the harvested energy. Moreover, the transmitter is subject to the physical power constraint
\begin{equation}
0\le N_s\le N_{\max},
\label{eq:Ns_budget}
\end{equation}
where $N_{\max}>0$ denotes the transmitter mean-photon-number budget determined by the source power. Constraint~\eqref{eq:Ns_budget} is the quantum analogue of the average-power constraint standard in classical Gaussian channel analysis and is consistent with the energy-constrained capacity framework of~\cite{Shirokov_2020}.}

The average harvested energy per channel use is modeled as
{
\begin{equation}
E_h = \left\langle \hat{H}_{\mathrm{EH}} \right\rangle
= \mathrm{Tr}\!\left[\hat{H}_{\mathrm{EH}} \, \rho_{\mathrm{EH}} \right],
\label{eq:EH_expectation}
\end{equation}}
{Here, $\hat{H}_{\mathrm{EH}}$ denotes the Hamiltonian associated with the EH output mode, and $\rho_{\mathrm{EH}}$ is the corresponding reduced density operator at the EH port}. {Specifically, under the generalized BS receiver of \eqref{eq:cv_BS_modes}, the average harvested energy per channel use is
\begin{equation}
E_h(N_s,\tau)=\eta_h(1-\tau)\,N_{\mathrm{out}}(N_s),
\label{eq:Eh_general}
\end{equation}
where $\eta_h\in(0,1]$ is the EH efficiency, and $N_{\mathrm{out}}(N_s)$ is the mean output photon number for the channel under consideration, specializing to $N_{\mathrm{out}}^{\mathrm{bos}}(N_s)$ or $N_{\mathrm{out}}^{\mathrm{AGN}}(N_s)$ in Sections~IV-B and~IV-D, respectively. The capacity-power function for a CV channel $\mathcal{N}$ is then defined as
\begin{equation}
C_Q(B)=\!\!\!\sup_{\substack{0\le N_s\le N_{\max}\\ 0\le\tau\le 1\\ E_h(N_s,\tau)\ge B}}\!\!\chi_{\mathrm{comm}}(N_s,\tau),
\label{eq:cv_capacity_power}
\end{equation}
where $\chi_{\mathrm{comm}}(N_s,\tau)$ denotes the Holevo information of the input ensemble transmitted through $\mathcal{N}$ and measured at the BS comm port. Since $[0,N_{\max}]\times[0,1]$ is compact and the admissible set defined by~\eqref{eq:cv_capacity_power} is closed, the supremum is attained  under standard continuity and compactness conditions, which hold in the considered setting.} The exact expression for $\chi_{\mathrm{comm}}$ depends on the encoding family (coherent, squeezed, or thermal), the channel transformation $\mathcal{N}$, and the symplectic spectrum of the output covariance matrix. {Throughout Sections~IV-B-IV-E, superscripts appended to $\chi_{\mathrm{comm}}$ and $C_Q$ identify the encoding family and channel: ``coh'' denotes coherent-state encoding, ``sq'' denotes displaced squeezed-state encoding; ``bos'' identifies the lossy bosonic channel (Section~IV-B) and ``AGN'' identifies the additive Gaussian noise channel (Section~IV-D).} We now apply the CV framework to specific single-mode Gaussian channels, beginning with the lossy bosonic channel, a key model for attenuation in optical quantum communication systems.

\subsection{Lossy bosonic channel with coherent state encoding}
The lossy bosonic channel $\mathcal{N}$ is a single-mode Gaussian channel that transforms the input annihilation operator $\hat{a}_{\mathrm{in}}$ according to the Heisenberg-picture relation as follows
\begin{equation}
\hat{a}_{\mathrm{out}} = \sqrt{T} \, \hat{a}_{\mathrm{in}} + \sqrt{1 - T} \, \hat{b}, \label{eq:bosonic_heisenberg}
\end{equation}
where $T \in [0,1]$ is the channel transmissivity, and $\hat{b}$ is an environmental mode in a thermal state with mean photon number $n_{\mathrm{env}}$. {The output mode $\hat{a}_{\mathrm{out}}$ is directed to the generalized BS receiver of Section~IV-A with transmissivity $\tau\in[0,1]$, producing the comm-port and EH-port operators in~\eqref{eq:cv_BS_modes}, where $\hat{v}$ denotes the vacuum ancilla. The mean photon number at the channel output is
\begin{equation}
{N_{\mathrm{out}}^{\mathrm{bos}}}(N_s)=TN_s+(1-T)n_{\mathrm{env}},
\label{eq:Nout_bosonic}
\end{equation}
so that the expected harvested energy per channel use becomes
\begin{equation}
E_h(N_s,\tau)=\eta_h(1-\tau){N_{\mathrm{out}}^{\mathrm{bos}}}(N_s).
\label{eq:bosonic_energy}
\end{equation}
The transmitter is constrained by the budget $N_s\le N_{\max}$. {Writing $N_{\mathrm{out}}^{\max,\mathrm{bos}}=TN_{\max}+(1-T)n_{\mathrm{env}}$ for the maximum output photon number under the transmitter budget, the maximum harvestable energy is $B_{\max}^{\mathrm{bos}}=\eta_h N_{\mathrm{out}}^{\max,\mathrm{bos}}$.}}

For coherent-state inputs $|\alpha\rangle$ with $\alpha\sim\mathcal{CN}(0,N_s)$, the comm-port output is a displaced thermal state and the average comm-port state is a zero-mean thermal state with mean photon number $\tau {N_{\mathrm{out}}^{\mathrm{bos}}}(N_s)$, whereas each conditional comm-port state is a displaced thermal state with conditional mean photon number $\tau(1-T)n_{\mathrm{env}}$. {Consequently, the comm-port Holevo information is given by
\begin{equation}
{\chi_{\mathrm{comm}}^{\mathrm{coh,bos}}}(N_s,\tau)=g\bigl(\tau {N_{\mathrm{out}}^{\mathrm{bos}}}(N_s)\bigr)-g\bigl(\tau(1-T)n_{\mathrm{env}}\bigr),
\label{eq:bosonic_holevo_comm}
\end{equation}
where $g(x)=(x+1)\log_2(x+1)-x\log_2 x$ is the bosonic entropy function. The capacity-power function is then obtained as the solution of the joint optimization in~\eqref{eq:cv_capacity_power}, which is characterized in closed form in Theorem~\ref{thm:bosonic_coh}. {Throughout the remainder of this paper, $\tau^\star(B)$ denotes the optimal BS transmissivity as a function of the EH requirement $B$; the scalar value $\tau^\star$ is used only when $B$ is fixed or at a boundary point.}}
\begin{theorem}\label{thm:bosonic_coh}
{Let ${N_{\mathrm{out}}^{\max,\mathrm{bos}}}=TN_{\max}+(1-T)n_{\mathrm{env}}$. The capacity-power function of the lossy bosonic channel under coherent-state encoding, with mean photon number budget $N_{\max}$ and jointly optimized BS transmissivity $\tau\in[0,1]$, is
\begin{equation}
{C_Q^{\mathrm{coh,bos}}}(B)
=g\bigl(\tau^\star(B)\,{N_{\mathrm{out}}^{\max,\mathrm{bos}}}\bigr)
-g\bigl(\tau^\star(B)(1-T)n_{\mathrm{env}}\bigr),
\label{eq:bosonic_cq2}
\end{equation}
for all $B\in[0,\eta_h {N_{\mathrm{out}}^{\max,\mathrm{bos}}}]$, where the optimal BS transmissivity and mean photon number are
\begin{equation}
\tau^\star(B)=1-\frac{B}{\eta_h {N_{\mathrm{out}}^{\max,\mathrm{bos}}}},\qquad N_s^\star=N_{\max}.
\label{eq:bosonic_taustar}
\end{equation}
The function ${C_Q^{\mathrm{coh,bos}}}(B)$ is strictly concave and strictly decreasing in $B$, with ${C_Q^{\mathrm{coh,bos}}}(0)=g({N_{\mathrm{out}}^{\max,\mathrm{bos}}})-g((1-T)n_{\mathrm{env}})$ and ${C_Q^{\mathrm{coh,bos}}}(\eta_h {N_{\mathrm{out}}^{\max,\mathrm{bos}}})=0$.}
\end{theorem}
\begin{IEEEproof}
The proof is presented in Appendix \ref{ProofofHelovo}.
\end{IEEEproof}
{The expression in~\eqref{eq:bosonic_cq2} provides an explicit closed-form characterization of the capacity-power trade-off under coherent-state encoding. The optimal strategy operates the transmitter at the full photon-number budget $N_s^\star=N_{\max}$, and realizes the trade-off entirely through the receiver-side BS transmissivity. As the EH requirement $B$ increases, $\tau^\star(B)$ decreases affinely, diverting an increasing fraction of the received optical power to the EH port and progressively reducing the photons reaching the information decoder. The capacity vanishes at the feasibility boundary $B=\eta_h {N_{\mathrm{out}}^{\max,\mathrm{bos}}}$, where $\tau^\star=0$ and the entire output is consumed for harvesting. Strict concavity of ${C_Q^{\mathrm{coh,bos}}}(B)$ follows from the composition of the affinely decreasing $\tau^\star(B)$ with the strictly concave auxiliary function $f(\tau):=g(\tau {N_{\mathrm{out}}^{\max,\mathrm{bos}}})-g(\tau(1-T)n_{\mathrm{env}})$, as established in Appendix~\ref{ProofofHelovo}.}
\subsection{Lossy bosonic channel with squeezed state encoding}
{Under the lossy bosonic channel, a displaced squeezed-vacuum input with squeezing parameter $r\ge 0$ is mapped to a Gaussian output state whose conditional quadrature variances are}
\begin{align}
\label{eq:42}
    \mathrm{Var}(\hat{X}_{\mathrm{out}}) &= \frac{1}{2} \left( T e^{-2r} + (1 - T)(2 n_{\mathrm{env}} + 1) \right), \\
    \label{eq:43}
    \mathrm{Var}(\hat{P}_{\mathrm{out}}) &= \frac{1}{2} \left( T e^{2r} + (1 - T)(2 n_{\mathrm{env}} + 1) \right).
\end{align}
{The mean photon number of the modulated input ensemble is}
\begin{equation}
N_s = \sinh^2(r) + \tfrac{1}{2}\!\left(\sigma_X^2 + \sigma_P^2\right),
\label{eq:sq_Ns}
\end{equation}
{where $\sinh^2(r)$ is the deterministic photon contribution of the squeezing vacuum and $\sigma_X^2$, $\sigma_P^2$ are the modulation variances along the $X$ and $P$ quadratures, respectively. The output mode $\hat{a}_{\mathrm{out}}$ then enters the generalized BS receiver of Section~IV-A, producing comm-port and EH-port covariance matrices that depend on both $(r,\sigma_X^2,\sigma_P^2)$ and $\tau$, as detailed in Appendix~\ref{ProofOfBosSq}. The mean output photon number $N_{\mathrm{out}}^{\mathrm{bos}}(N_s)=TN_s+(1-T)n_{\mathrm{env}}$ and the harvested energy $E_h(N_s,\tau)=\eta_h(1-\tau)N_{\mathrm{out}}^{\mathrm{bos}}(N_s)$ take the same form as in Section~IV-B. Unlike the coherent case, however, the achievable Holevo information at the comm port depends on all three ensemble parameters $(r,\sigma_X^2,\sigma_P^2)$, since squeezing introduces quadrature asymmetry in the output covariance. The result is summarized in the following proposition.}
\begin{proposition}\label{prop:bosonic_sq}
{The achievable-rate–power function of the lossy bosonic channel under displaced squeezed-state encoding, with mean photon number budget $N_{\max}$ and jointly optimized BS transmissivity $\tau\in[0,1]$, is
\begin{equation}
{C_Q^{\mathrm{sq,bos}}}(B)=\!\!\!\!\sup_{\substack{r\ge 0,\;\sigma_X^2,\sigma_P^2\ge 0\\
N_s\le N_{\max},\;0\le\tau\le 1\\
\eta_h(1-\tau){N_{\mathrm{out}}^{\mathrm{bos}}}(N_s)\ge B}}\!\!\!\!
{\chi_{\mathrm{comm}}^{\mathrm{sq,bos}}}(r,\sigma_X^2,\sigma_P^2,\tau),
\label{eq:cv_sq_optim}
\end{equation}
where ${\chi_{\mathrm{comm}}^{\mathrm{sq,bos}}}(r,\sigma_X^2,\sigma_P^2,\tau)$ is the comm-port Holevo information given in Appendix~\ref{ProofOfBosSq}, and $N_s=\sinh^2(r)+(\sigma_X^2+\sigma_P^2)/2$.}
\end{proposition}
\begin{IEEEproof}
The proof is presented in Appendix \ref{ProofOfBosSq}
\end{IEEEproof}
Unlike the coherent case, this achievable-rate-power function does not have a closed-form expression. However, numerical optimization over the ensemble parameters \( (r, \sigma_X^2, \sigma_P^2) \) can be used to compute lower bounds. The entropy terms are evaluated through the symplectic eigenvalues of the corresponding covariance matrices, as detailed in Appendix~\ref{ProofOfBosSq}.
\begin{remark}[Comparison to coherent encoding]
\label{re:comp:bos}
{Although squeezed-vacuum states are included for completeness, our numerical results show that they do not outperform coherent-state encoding under the considered QSIPT model. This behavior holds across all examined noise and EH thresholds,  where coherent states achieve the Holevo capacity when no entanglement or collective measurements are employed \cite{Giovani, Holevo2001, Leditzky_2018}.  The reason is that, for phase-insensitive bosonic channels, squeezing applied to the vacuum redistributes quadrature noise without increasing the classical distinguishability of displaced codewords,  while simultaneously consuming photons that reduce the modulation budget $N_{\mathrm{mod}} := \tfrac12(\sigma_X^2+\sigma_P^2)= N_s - \sinh^2(r)$, which strictly decreases the comm-port Holevo information under the phase-insensitive channel and passive beam-splitter receiver. Consequently, coherent-state encoding achieves the best performance among the considered displaced Gaussian inputs. We note that potential advantages of squeezing may arise in phase-sensitive channels or under joint modulation squeezing strategies, which are beyond the scope of this work.}
\end{remark}
\begin{remark}[{Feasibility of the EH constraint under squeezed-state encoding}]\label{rem:rstar}
{Under the revised generalized BS receiver of Section~IV-A, the EH constraint $\eta_h(1-\tau)N_{\mathrm{out}}^{\mathrm{bos}}(N_s)\ge B$ associated with the squeezed-state optimization in~\eqref{eq:cv_sq_optim} is always satisfiable for any $B\in[0,B_{\max}^{\mathrm{bos}}]$ via $\tau\in[0,\tau_{\max}(B)]$ with $\tau_{\max}(B)=1-B/(\eta_h N_{\mathrm{out}}^{\max,\mathrm{bos}})\ge 0$, and does \emph{not} impose any additional lower bound on the modulation budget $N_{\mathrm{mod}}$. Consequently, the squeezed-state achievable-rate-power curve $C_Q^{\mathrm{sq,bos}}(B)$ does \emph{not} drop prematurely to zero before reaching $B_{\max}^{\mathrm{bos}}$; it vanishes only at the same feasibility boundary as the coherent-state curve. The performance gap between the two encodings is therefore attributable entirely to the structural reasons identified in Remark~5, with no infeasibility of the EH constraint.}
\end{remark}

\subsection{Lossy bosonic channel with thermal state encoding}
 Under the lossy bosonic channel $\mathcal{N}$, a thermal input state defined in \eqref{eqRhoTh} with mean photon number $N_s$ is mapped to another thermal state with output photon number given by
\begin{equation}
{N_{\mathrm{out}}^{\mathrm{bos}}}(N_s) = T N_s + (1 - T) n_{\mathrm{env}}.
\end{equation}

Since thermal states are preserved under Gaussian attenuation, the output remains diagonal in the Fock basis with entropy
\begin{equation}
S(\mathcal{N}(\rho_{\text{th}})) = g\bigl({N_{\mathrm{out}}^{\mathrm{bos}}}(N_s)\bigr).
\end{equation}
In this work, thermal-state encoding is used only as an energy-harvesting baseline, with no classical message encoded in the thermal energy; hence the corresponding information-bearing ensemble is trivial and its achievable communication rate is zero. Therefore,
\begin{equation}
C_Q^{\rm th}(B)=0
\end{equation}
for all feasible $B \geq 0$.
Because all thermal input states map to thermal output states of the same functional form, no classical information can be encoded in their mixture.

Hence, the achievable-rate-power function under thermal encoding satisfies
\begin{equation}
\label{eq:CQ_th}
C_Q^{\mathrm{th}}(B) = 0 \quad \text{for all feasible } B\ge 0,
\end{equation}
a result that holds for both the lossy bosonic and AGN channels, since the argument depends only on the phase-insensitive and diagonal nature of thermal states, which is channel-independent.

{Since $C_Q^{\mathrm{th}}(B)=0$ irrespective of $\tau$, the EH constraint $E_h\ge B$ is satisfied by setting $\tau=0$ and routing the entire channel output to the EH port, yielding the maximum harvestable energy $\eta_h\,{N_{\mathrm{out}}^{\mathrm{bos}}}(N_s)$, which can be tuned by increasing $N_s$.}

This highlights the essential trade-off in QSIPT: energy and information transfer require structured input ensembles. In particular, coherent and squeezed states enable distinguishability through phase-space structure, whereas thermal states lack this property. Although thermal encoding maximizes the entropy per photon, it is unsuitable for classical information transfer due to its inherently mixed and isotropic nature.
\begin{remark}[Thermal Encoding as a Power-Only Scheme]
Thermal states are the optimal input to maximize entropy at the output for a given energy. However, thermal states are unsuitable for information transfer because they are diagonal in the Fock basis and lack phase coherence, making them incapable of encoding distinguishable signals. In QSIPT systems, thermal encoding may still be valuable in “pure EH”, where information transfer is not required or as part of a time-sharing strategy with coherent states to balance energy and information transfer objectives.
\end{remark}

\subsection{AGN channel with coherent state encoding}
The AGN channel is a type of quantum Gaussian channel in which a quantum signal is subjected to Gaussian-distributed noise that is independent of the input state. This model is especially relevant for scenarios where the main source of degradation is the background noise rather than loss or attenuation. In this channel, the quadratures of the input state are affected by Gaussian noise, resulting in output quadratures that are represented as
{ \begin{align}     \hat{x}_{\mathrm{out}} 
&= \hat{x}_{\mathrm{in}} + \hat{x}_{\mathrm{noise}}, \\
\hat{p}_{\mathrm{out}} 
&= \hat{p}_{\mathrm{in}} + \hat{p}_{\mathrm{noise}}, \\
\langle \hat{x}_{\mathrm{noise}}^2 \rangle 
&= \langle \hat{p}_{\mathrm{noise}}^2 \rangle 
= \frac{2N_0+1}{2}, \\
\bigl[\hat{x}_{\mathrm{noise}}, \hat{p}_{\mathrm{noise}}\bigr] 
&= i.
\end{align}
}
where \( \hat{x}_{\text{in}} \) and \( \hat{p}_{\text{in}} \) represent the input quadratures of the signal;
  \( \hat{x}_{\text{noise}} \) and \( \hat{p}_{\text{noise}} \) represent Gaussian noise terms with mean zero and variance \( N_0 \), independently added to each quadrature.
{The output mode $\hat{a}_{\mathrm{out}}$ enters the generalized BS receiver of Section~IV-A with transmissivity $\tau\in[0,1]$. Under the transmitter photon-number budget $N_s\le N_{\max}$, the mean photon number at the channel output is $N_{\mathrm{out}}^{\mathrm{AGN}}(N_s)=N_s+N_0$, and the expected harvested energy is
\begin{equation}
\label{eqEHAGN}
E_h(N_s,\tau)=\eta_h(1-\tau)(N_s+N_0).
\end{equation}
Writing $N_{\mathrm{out}}^{\max,\mathrm{AGN}}:=N_{\max}+N_0$, the maximum harvestable energy under the photon-number budget is $B_{\max}^{\mathrm{AGN}}=\eta_h(N_{\max}+N_0)$.} We consider coherent state encoding as described in Section~II-B, where the input consists of coherent states $|\alpha\rangle$ drawn from a circularly symmetric complex Gaussian distribution with zero mean and variance $N_s$. {For coherent-state inputs, the comm-port Holevo information is
\begin{equation}
\chi_{\mathrm{comm}}^{\mathrm{coh,AGN}}(N_s,\tau)=g\bigl(\tau(N_s+N_0)\bigr)-g(\tau N_0),
\label{eq:AGN_holevo_comm}
\end{equation}
which reduces to the unconstrained Holevo capacity $g(N_{\max}+N_0)-g(N_0)$ at $\tau=1$ (no EH), and to zero at $\tau=0$ (all received power diverted to EH). The capacity-power function is the solution of the joint optimization in~\eqref{eq:cv_capacity_power} specialized to this channel, and admits the closed-form characterization stated next.}

\begin{theorem}\label{thm:AGN_coh}
{Let $N_{\mathrm{out}}^{\max,\mathrm{AGN}}:=N_{\max}+N_0$. The capacity-power function of the AGN channel under coherent-state encoding, with mean photon number budget $N_{\max}$ and jointly optimized BS transmissivity $\tau\in[0,1]$, is
\begin{equation}
\label{eq:AGN_CQ}
C_Q^{\mathrm{coh,AGN}}(B)=g\bigl(\tau^\star(B)(N_{\max}+N_0)\bigr)-g\bigl(\tau^\star(B)\,N_0\bigr),
\end{equation}
for all $B\in[0,\eta_h(N_{\max}+N_0)]$, where the optimal BS transmissivity and mean photon number are
\begin{equation}
\tau^\star(B)=1-\frac{B}{\eta_h(N_{\max}+N_0)},\qquad N_s^\star=N_{\max}.
\label{eq:AGN_taustar}
\end{equation}
The function $C_Q^{\mathrm{coh,AGN}}(B)$ is strictly concave and strictly decreasing in $B$, with $C_Q^{\mathrm{coh,AGN}}(0)=g(N_{\max}+N_0)-g(N_0)$ and $C_Q^{\mathrm{coh,AGN}}(\eta_h(N_{\max}+N_0))=0$.}
\end{theorem}
{
\begin{IEEEproof}
The proof follows the same two-lemma structure as Appendix~\ref{ProofofHelovo},
with $N_{\mathrm{out}}^{\mathrm{AGN}}(N_s)=N_s+N_0$, $T=1$, and $(1-T)n_{\mathrm{env}}$
replaced by $N_0$. Specifically, with $f_{\mathrm{AGN}}(\tau):=g\bigl(\tau(N_{\max}+N_0)\bigr)-g(\tau N_0)$,
Lemma~D.1 gives $N_s^\star=N_{\max}$, and the analogue of Lemma~D.2 yields
\[
f_{\mathrm{AGN}}'(\tau)
=\frac{1}{\tau}\!\left[\varphi\bigl(\tau(N_{\max}+N_0)\bigr)-\varphi(\tau N_0)\right]>0,
\]
since $N_{\max}+N_0>N_0$. Strict concavity follows from
\[
f_{\mathrm{AGN}}''(\tau)
=\frac{N_0-(N_{\max}+N_0)}{\tau\ln 2\,\bigl(\tau(N_{\max}+N_0)+1\bigr)(\tau N_0+1)}<0.
\]
The remainder of the proof is identical to Appendix~\ref{ProofofHelovo}.
\end{IEEEproof}
}
{The expression in~\eqref{eq:AGN_CQ} reveals a structure parallel to that of the lossy bosonic channel under coherent-state encoding (Theorem~\ref{thm:bosonic_coh}), with multiplicative attenuation $T$ replaced by additive noise $N_0$. The capacity decreases monotonically and concavely from its unconstrained value $g(N_{\max}+N_0)-g(N_0)$ at $B=0$ to zero at $B=\eta_h(N_{\max}+N_0)$, where the entire received output is routed to the EH port. The AGN channel is less sensitive to input state structure, and coherent states remain optimal among the considered Gaussian encodings under average energy constraints.}
\begin{remark}[Comparison with the lossy bosonic channel]
{The structural similarity between $C_Q^{\mathrm{coh,bos}}(B)$ and $C_Q^{\mathrm{coh,AGN}}(B)$ reflects the parallel roles of multiplicative attenuation and additive noise in constraining the comm-port Holevo information, as further examined in Fig.~7.}
However, unlike the lossy channel, where attenuation scales both signal and noise, the AGN model assumes constant additive noise and linear EH, making it a useful benchmark for QSIPT system design.
\end{remark}
\subsection{AGN channel with squeezed state encoding}
{We consider a displaced squeezed-vacuum input ensemble as in Section~IV-C, with the AGN channel of Section~IV-D replacing the lossy bosonic channel. The channel adds phase-insensitive Gaussian noise of variance $N_0$ to both quadratures, preserving the asymmetric quadrature structure of the squeezed vacuum. The conditional output quadrature variances (prior to the generalized BS receiver of Section~IV-A) are
\begin{align}
\mathrm{Var}(\hat{X}_{\mathrm{out}}) &= \tfrac{1}{2} e^{-2r} + N_0, \\
\mathrm{Var}(\hat{P}_{\mathrm{out}}) &= \tfrac{1}{2} e^{2r} + N_0,
\end{align}
where $r\ge 0$ is the squeezing parameter. Note that unlike the lossy bosonic channel, the AGN channel does not attenuate the displacement modulation, so the modulation variances $\sigma_X^2,\sigma_P^2$ enter the average output covariance with unit coefficient before the BS. The generalized BS with transmissivity $\tau\in[0,1]$ further attenuates the comm-port output by $\tau$ and introduces vacuum noise $(1-\tau)/2$ per quadrature (see Appendix~\ref{ProofOfAGNSQ}).

The transmitter budget $N_s\le N_{\max}$ is enforced, with
\begin{equation}
N_s = \sinh^2(r) + \tfrac{1}{2}(\sigma_X^2+\sigma_P^2)\le N_{\max},
\end{equation}
and the EH constraint at the EH port reads $\eta_h(1-\tau)(N_s+N_0)\ge B$. The achievable-rate-power function under displaced squeezed-state encoding is obtained via the following proposition.}
\begin{proposition}\label{prop:AGN_sq}
{Under the generalized BS receiver of Section~IV-A with transmissivity $\tau\in[0,1]$ and transmitter budget $N_s\le N_{\max}$, the achievable-rate-power function of the AGN channel under displaced squeezed-state encoding is
\begin{equation}\label{eq:AGN_sq_optim}
\begin{aligned}
C_Q^{\mathrm{sq,AGN}}(B)
=
\sup_{\substack{
r\ge 0,\;\sigma_X^2,\sigma_P^2\ge 0,\;\tau\in[0,1]\\
N_s\le N_{\max}\\
\eta_h(1-\tau)(N_s+N_0)\ge B
}}
\;
\chi_{\mathrm{comm}}^{\mathrm{sq,AGN}}
(r,\sigma_X^2,\sigma_P^2,\tau).
\end{aligned}
\end{equation}
where $\chi_{\mathrm{comm}}^{\mathrm{sq,AGN}}$ is the comm-port Holevo information evaluated from the symplectic eigenvalues of the average and conditional comm-port covariance matrices (see Appendix~\ref{ProofOfAGNSQ}). The optimization does not admit a closed form for $r>0$ and is performed numerically. Numerical results over a representative grid of parameters consistently yield the optimal squeezing $r^\star=0$, so that displaced coherent-state encoding (Theorem~\ref{thm:AGN_coh}) achieves the capacity-power function within the displaced Gaussian encoding class.}
\end{proposition}
\begin{IEEEproof}
The proof is presented in Appendix~\ref{ProofOfAGNSQ}.
\end{IEEEproof}
\begin{remark}[Optimality of coherent states for the AGN channel]
{The structural reason for $r^\star=0$ in the AGN setting mirrors that of the lossy bosonic channel (Remark~\ref{re:comp:bos}): the phase-covariance of the channel-BS cascade implies that the comm-port Holevo information depends only on the isotropic component of the input quadrature variances, so quadrature asymmetry introduced by squeezing provides no gain. Simultaneously, any $r>0$ strictly reduces the modulation budget $N_{\mathrm{mod}}=N_s-\sinh^2(r)$ at the same total photon number $N_s=N_{\max}$, which strictly decreases the comm-port Holevo information. Coherent-state encoding  is therefore strictly optimal within the considered displaced Gaussian encoding class.}
\end{remark}
\subsection{AGN with thermal state encoding}
We consider a thermal input state defined in \eqref{eqRhoTh} with mean photon number $N_s$. The AGN channel adds phase-insensitive Gaussian noise to both quadratures
{\begin{align}
\hat{a}_{\mathrm{out}} 
&= \hat{a}_{\mathrm{in}} + \hat{a}_{\mathrm{noise}}, \label{eq:agn_operator} \\
\langle \hat{a}_{\mathrm{noise}} \rangle 
&= 0, \qquad 
\langle \hat{a}_{\mathrm{noise}}^{\dagger} \hat{a}_{\mathrm{noise}} \rangle = N_0, \\
\bigl[\hat{a}_{\mathrm{noise}}, \hat{a}_{\mathrm{noise}}^{\dagger}\bigr] 
&= 1
\end{align}}

Since thermal states are phase-insensitive, and the AGN channel acts symmetrically, the output is again a thermal state with mean photon number ${N_{\mathrm{out}}^{\mathrm{AGN}}}(N_s) = N_s + N_0$. Since thermal states are diagonal and commute, any ensemble of thermal inputs maps to output states of identical spectral shape. As in the lossy bosonic case, the thermal curve is treated as an energy-only baseline, with no classical message encoded in the thermal energy. Under this convention, the information-bearing ensemble is trivial and the achievable communication rate is zero, i.e., 
\begin{equation}
\label{eq:CQ_th_AGN}
C_Q^{\mathrm{th}}(B) = 0 \quad \text{for all feasible } B\ge 0,
\end{equation}
consistently with~\eqref{eq:CQ_th}. {Since $C_Q^{\mathrm{th}}(B)=0$ irrespective of $\tau$, the EH constraint $E_h\ge B$ is satisfied by setting $\tau=0$, which yields the feasibility condition $N_s\ge B/\eta_h - N_0$.}
This constraint is always feasible for sufficiently large $N_s$, even though the corresponding $C_Q(B)$ remains zero.
Thermal-state encoding over the AGN channel illustrates the complete decoupling between energy and information in quantum optical systems: these states maximize output entropy per photon but lack phase coherence, rendering them incapable of encoding distinguishable classical messages.
\begin{remark}[Thermal States as EH-Only Inputs]
Thermal states serve as “power-only” signals that saturate the EH constraint without contributing to communication. In QSIPT systems, they may be used in dedicated EH phases or in time-sharing with coherent encodings to optimize the achievable-rate-power function.
\end{remark}
{The preceding analysis considers the capacity-power trade-off of coherent-state encoding and the achievable-rate–power trade-off of the squeezed and thermal-state encodings. More generally, composite signaling strategies can be employed to enlarge the achievable information-energy region.}
{\subsection{Time-Sharing and Hybrid Encoding Strategies}}
{
\subsubsection{Time-Sharing Between Coherent and Thermal States}
    We first consider a time-sharing strategy in which the transmitter allocates a fraction $\alpha \in [0,1]$ of channel uses to coherent-state encoding and the remaining fraction $1-\alpha$ to thermal-state transmission. The capacity-power trade-off of coherent-state encoding is characterized by the coherent-state capacity-power function ${C_Q^{\mathrm{coh,bos}}}$ derived in Section~IV-B. {Let $B_{\mathrm{coh}}\in[0,B_{\max}^{\mathrm{bos}}]$ denote the per-use EH target under coherent-state encoding, $B_{\mathrm{th}}\in[0,B_{\max}^{\mathrm{bos}}]$ the per-use EH target under thermal-state encoding, and $\mathcal{B}:=[0,B_{\max}^{\mathrm{bos}}]\times[0,B_{\max}^{\mathrm{bos}}]$ the joint feasible EH set.} During the coherent transmission phase, information is conveyed at rate ${C_Q^{\mathrm{coh,bos}}}(B_{\mathrm{coh}})$ while delivering harvested energy $B_{\mathrm{coh}}$ per channel use. During the thermal phase, no information is transmitted, but a higher harvested energy level $B_{\mathrm{th}}$, corresponding to thermal-state transmission, can be delivered.}
     {
\begin{equation}
\bigl(C,B\bigr)
=
\Bigl(
\alpha\, {C_Q^{\mathrm{coh,bos}}}(B_{\mathrm{coh}}),
\;
\alpha B_{\mathrm{coh}} + (1-\alpha) B_{\mathrm{th}}
\Bigr),
\label{eq:timesharing_region}
\end{equation}
}
{
for all feasible $\alpha \in [0,1]$ and $(B_{\mathrm{coh}},B_{\mathrm{th}}) \in \mathcal{B}$, where $\mathcal{B}$ denotes the set of harvested energy levels achievable under coherent and thermal encoding, respectively. This construction yields the convex hull of the achievable points obtained with coherent and thermal states alone. As a result, time-sharing  enlarges the achievable-rate-power region, particularly in the high EH regime where purely coherent encoding suffers from severe capacity degradation.}
{
\subsubsection{Hybrid Encoding via Mixed Quantum States}
Beyond time-sharing across channel uses, we next consider a hybrid encoding strategy in which the transmitter employs a convex mixture of information-carrying and energy-carrying input states within a single channel use. Specifically, we consider input ensembles described by mixed quantum states of the form
\begin{equation}
\rho = p\, \rho_{\mathrm{coh}} + (1-p)\, \rho_{\mathrm{th}},
\end{equation}
where $\rho_{\mathrm{coh}}$ denotes the ensemble-average input state induced by coherent-state encoding, $\rho_{\mathrm{th}}$ denotes a thermal input state optimized for energy harvesting as characterized in Remark~10, and $p \in [0,1]$ controls the mixing ratio. For such hybrid inputs, the EH constraint couples linearly to $p$, while the achievable information rate depends nonlinearly on the effective photon budget allocated to coherent modulation. Let $C_Q^{\mathrm{hybrid,bos}}(B)$ denote the achievable-rate-power  function under hybrid mixed-state encoding over the lossy bosonic channel. Then, the following lower bound holds:}
{
\begin{equation}
C_Q^{\mathrm{hybrid,bos}}(B)
\geq
\max_{p \in [0,1]}
\left\{
p\, {C_Q^{\mathrm{coh,bos}}}\!\left(\frac{B-(1-p)B_{\mathrm{th}}}{p}\right)
\right\},
\label{eq:hybrid_bound}
\end{equation}
whenever the per-coherent-use harvested energy $\bigl(B-(1-p)B_{\mathrm{th}}\bigr)/p$
lies in $[0,B_{\max}^{\mathrm{bos}}]$, i.e., whenever
$B\ge(1-p)B_{\mathrm{th}}$ and $B\le(1-p)B_{\mathrm{th}}+p\,B_{\max}^{\mathrm{bos}}$,
ensuring both non-negativity of the argument and admissibility within the
domain of $C_Q^{\mathrm{coh,bos}}(\cdot)$. Due to the strict concavity of ${C_Q^{\mathrm{coh,bos}}}(B)$ established in Theorem~\ref{thm:bosonic_coh}, this hybrid strategy can outperform conventional time-sharing in certain EH regimes, yielding higher achievable information rates under the same EH requirement.}

\section{Numerical Analysis}
\subsection{DV case}
Fig.~3 illustrates the capacity-power function for the amplitude damping channel under an EH constraint. We consider three representative input states: the ground state $|0\rangle$ with energy $E_0=0$, the excited state $|1\rangle$ with energy $E_1=1$, and the equal superposition $\tfrac{1}{\sqrt{2}}(|0\rangle+|1\rangle)$ with mean energy $\mathrm{Tr}[H\rho]=1/2$, under the standard qubit normalization of Section~II-A. The plot shows a gradual decrease in capacity (in bits per channel use) as the EH requirement $B$ increases, reflecting the fundamental trade-off between information rate and harvested energy. For low $B$, the EH constraint is inactive and the unconstrained Holevo capacity is achieved. Beyond a critical threshold, the capacity decreases as the optimal input ensemble is progressively biased toward energy-carrying states.

Fig. 4 {depicts the capacity-power function $C_Q(B)$ of the quantum erasure channel under the standard qubit normalization $E_0=0$, $E_1=1$, with representative erasure probability $p_e=0.1$ and erasure-flag energy $E_e=0$. The curve exhibits the two regimes predicted by Theorem~2. For $B\le B^\star=(1-p_e)/2+p_eE_e$, the EH constraint is inactive and $C_Q(B)=1-p_e$. For $B^\star<B\le B_{\max}=(1-p_e)+p_eE_e$, the constraint is active and
\begin{equation*}
C_Q(B)=(1-p_e)H_2\!\left(\frac{B-p_eE_e}{1-p_e}\right).
\end{equation*}
This active-regime branch is strictly concave and monotonically decreasing. It leaves the plateau with zero slope at $B^\star$, since $H_2(\cdot)$ is maximized at $1/2$, and becomes increasingly steep as $B\to B_{\max}^-$. The figure therefore makes explicit that the EH requirement is met without rate loss below $B^\star$, whereas above $B^\star$ the input must be progressively biased toward the excited state at a corresponding loss in classical distinguishability.}
\subsection{CV case}
Fig.~5 {depicts the capacity-power and achievable-rate–power trade-offs $C_Q(B)$ for the lossy bosonic channel under three encoding families, with simulation parameters $T=0.8$, $n_{\mathrm{env}}=0.5$, $\eta_h=0.9$, and $N_{\max}=5$ photons/mode. The coherent-state curve, characterized in closed form by Theorem~\ref{thm:bosonic_coh} as the channel capacity, is strictly concave and monotonically decreasing throughout $B\in[0,\eta_h {N_{\mathrm{out}}^{\max,\mathrm{bos}}}]$, vanishing at the feasibility boundary where $\tau^\star\to 0$. The squeezed-state curve at $r=1.0$ lies strictly below the coherent-state baseline over the \emph{entire} feasible interval $[0,B_{\max}^{\mathrm{bos}}]$, reaching zero only at the same boundary $B_{\max}^{\mathrm{bos}}$. Under the generalized BS receiver with jointly optimized $\tau\in[0,1]$, the EH constraint is always satisfiable at any $B\in[0,B_{\max}^{\mathrm{bos}}]$ by adjusting $\tau$, and the performance gap between squeezed and coherent encodings is attributed solely to the reduction of the modulation budget $N_{\mathrm{mod}}=N_{\max}-\sinh^2(r)$  induced by squeezing, consistent with Remark 5. Thermal states yield $C_Q(B)=0$ throughout, confirming their role as energy-only carriers (Remark 7).}

Fig.~6 {shows the capacity-power trade-off for the AGN channel, characterized in Theorem~\ref{thm:AGN_coh} under the same transmitter budget and EH efficiency, with $N_0=0.5$ photons/mode. The coherent-state curve exhibits the same qualitative structure as in Fig.~5, i.e., strict concavity and monotone decrease via $\tau^\star(B)$ with two structural differences attributable to the additive-noise mechanism. First, the unconstrained Holevo capacity at $B=0$ is strictly lower than in the lossy bosonic case, since additive noise contributes an irreducible background entropy $g(N_0)$ at the communication port independent of the input power. Second, the feasibility horizon $B_{\max}^{\mathrm{AGN}}=\eta_h(N_{\max}+N_0)$ extends beyond its bosonic counterpart $B_{\max}^{\mathrm{bos}}=\eta_h {N_{\mathrm{out}}^{\max,\mathrm{bos}}}$, as the additive-noise photons contribute to the EH port without requiring a reduction of the transmitter power. The squeezed-state curve again lies strictly below the coherent-state baseline, consistent with the structural analysis of Section~IV-C, while the thermal-state curve forms a horizontal line at $C_Q(B)=0$, reflecting the complete decoupling between EH and information transfer in maximally mixed states (Remark 10).}

Fig.~7 {compares the coherent-state capacity-power functions of the lossy bosonic and AGN channels under matched parameters. Two analytically distinct effects govern the relative ordering of the curves. At low EH demand, the lossy bosonic curve dominates: although multiplicative attenuation reduces the total output photon number, it attenuates the noise floor proportionally, yielding a higher net Holevo information than the AGN channel, whose additive noise contributes an irreducible entropy offset at the communication port independent of the input power. At high EH demand, the ordering reverses: the additive noise photons of the AGN channel contribute directly to the harvestable energy without consuming any portion of the transmitter budget $N_{\max}$, extending the feasibility horizon beyond that of the lossy bosonic channel. Since the lossy bosonic curve is strictly higher at $B=0$ and vanishes strictly earlier, while both functions are continuous and strictly concave (Theorems~\ref{thm:bosonic_coh} and~\ref{thm:AGN_coh}), implies the existence of at least one crossing in the common feasible range, and the numerical evaluation in Fig. 7 shows a single interior crossing  at $B^\star\approx 2.17$ photons/mode. This crossing partitions the feasible region into a rate-limited regime ($B<B^\star$, lossy bosonic preferred) and an energy-limited regime ($B>B^\star$, AGN preferred), constituting a structural signature of the fundamental distinction between multiplicative attenuation and additive noise in the QSIPT setting.}

{The key difference between the DV and CV results is the significantly sharper reduction of the capacity-power function observed in DV channels as the EH requirement increases. This behavior is not specific to the particular DV channels considered, but instead reflects a fundamental distinction between finite and infinite dimensional quantum systems. In DV channels, communication occurs over a finite-dimensional Hilbert space with discrete energy eigenvalues, and the capacity-power optimization is restricted to finite input ensembles. As the EH constraint tightens, the optimal ensemble may undergo abrupt structural changes, shifting from capacity-optimal mixtures toward energy-biased signaling strategies, which leads to steep drops or plateaus in the capacity-power curve. In contrast, CV channels admit a continuous energy spectrum and allow the optimal input states to be smoothly parameterized by the average photon number. Therefore, increasing the EH requirement induces continuous adjustments in the optimal input parameters, resulting in a smoother and typically differentiable capacity-power trade-off. This distinction highlights the role of dimensionality and energy quantization in shaping the information-energy trade-off.}

\begin{figure}[t!]
    \centering
    \includegraphics[width=\linewidth]{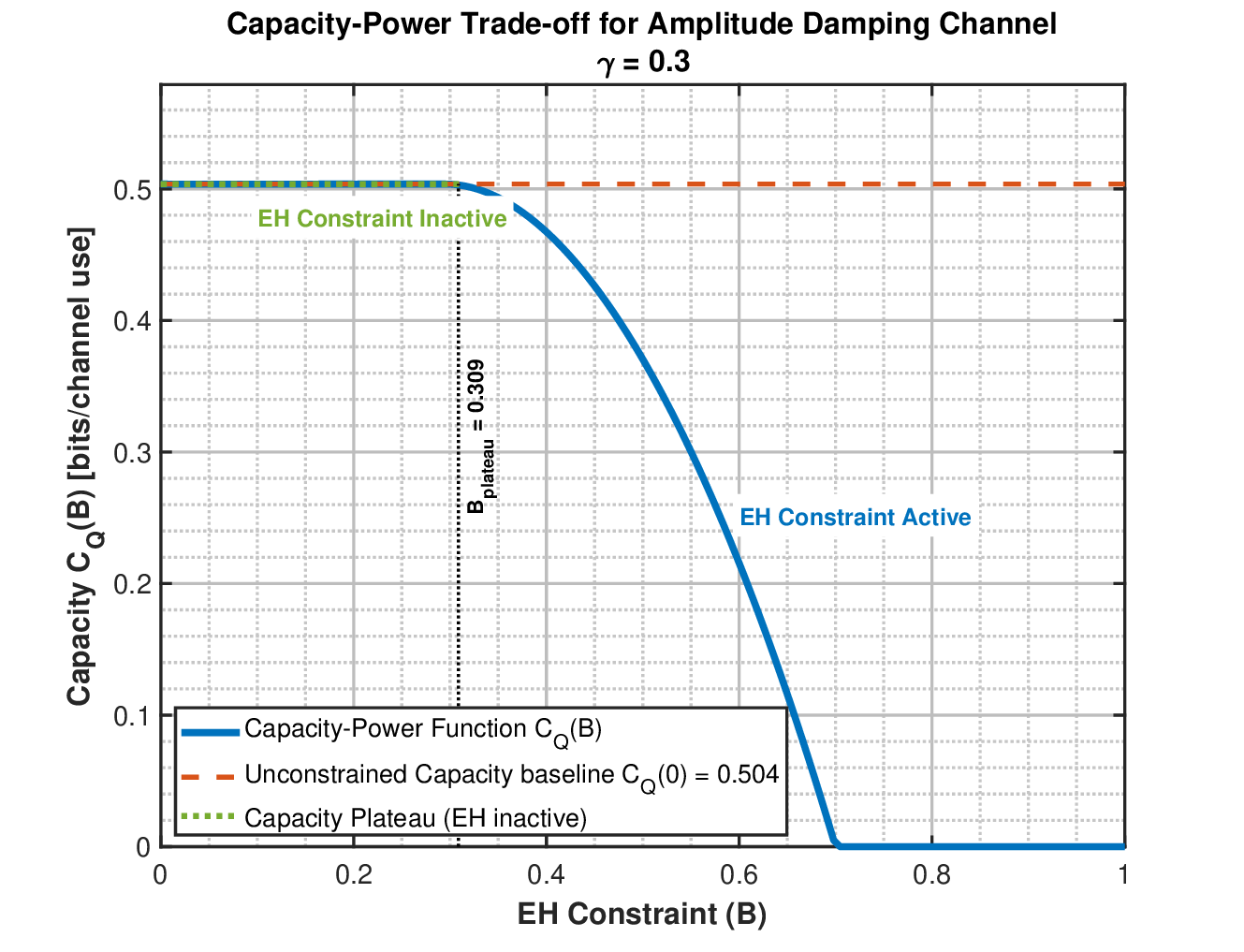}
    \caption{{Capacity-power function $C_Q(B)$ of the amplitude damping channel with damping parameter $\gamma=0.3$, under the standard qubit energy normalization $E_0=0$, $E_1=1$.}}
    \label{Damp}
\end{figure}

\begin{figure}[t!]
    \centering
    \includegraphics[width=\linewidth]{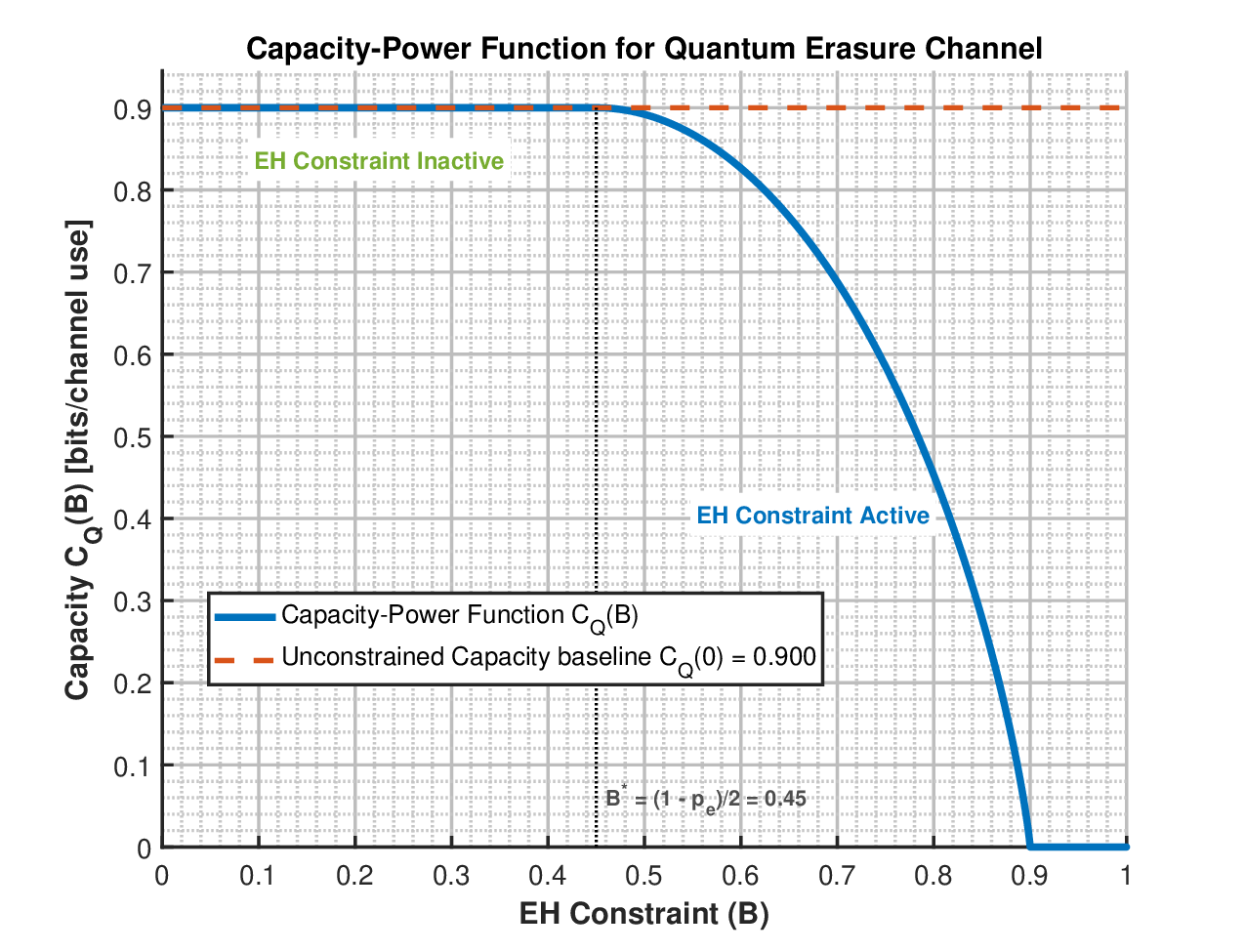}
    \caption{{Capacity-power function $C_Q(B)$ of the quantum erasure channel under $E_0=0$, $E_1=1$, $E_e=0$, and representative $p_e=0.1$. The threshold $B^\star=\frac{1-p_e}{2}+p_eE_e=0.45$ separates the inactive-EH plateau $C_Q(B)=1-p_e$ from the strictly concave active-EH regime.}}
    \label{ErasureFig}
\end{figure}

\begin{figure}[t!]
    \centering
    \includegraphics[width=0.5\textwidth]{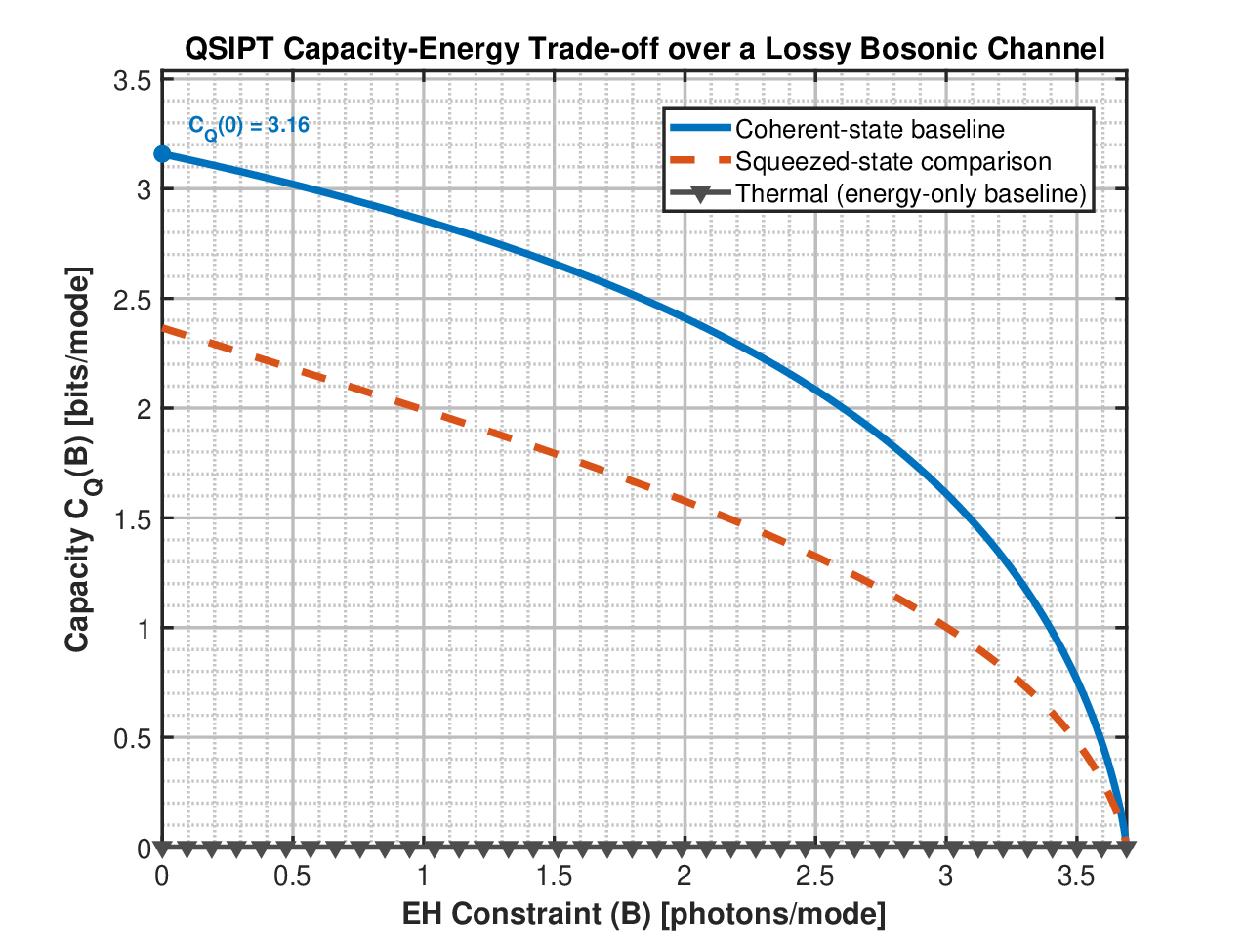}
    \caption{{Capacity-power trade-off for the lossy bosonic channel ($T=0.8$, $n_{\mathrm{env}}=0.5$, $\eta_h=0.9$, $N_{\max}=5$): coherent-state baseline (Theorem~\ref{thm:bosonic_coh}), displaced squeezed-state envelope at $r=1.0$ (Proposition~\ref{prop:bosonic_sq}), and thermal-state baseline.}}
    \label{fig:squeezed}
\end{figure}

\begin{figure}[t!]
    \centering
    \includegraphics[width=0.5\textwidth]{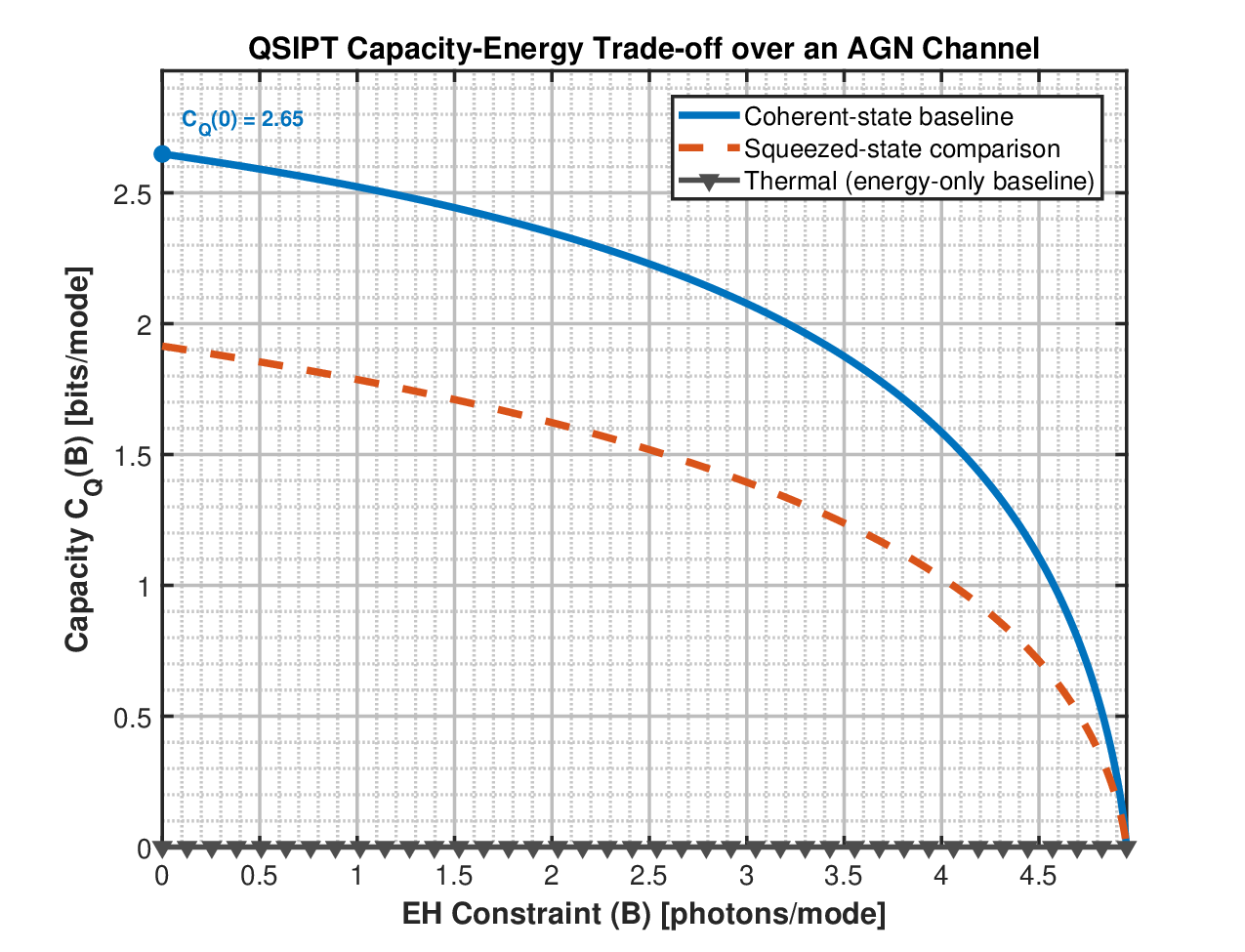}
     \caption{{Capacity-power trade-off for the AGN channel under $N_{\max}=5$, $\eta_h=0.9$, $N_0=0.5$: coherent-state baseline (Theorem~\ref{thm:AGN_coh}), displaced squeezed-state envelope at $r=1.0$, and thermal-state baseline.}}
    \label{fig:coherent}
\end{figure}

\begin{figure}[t!]
    \centering
    \includegraphics[width=0.5\textwidth]{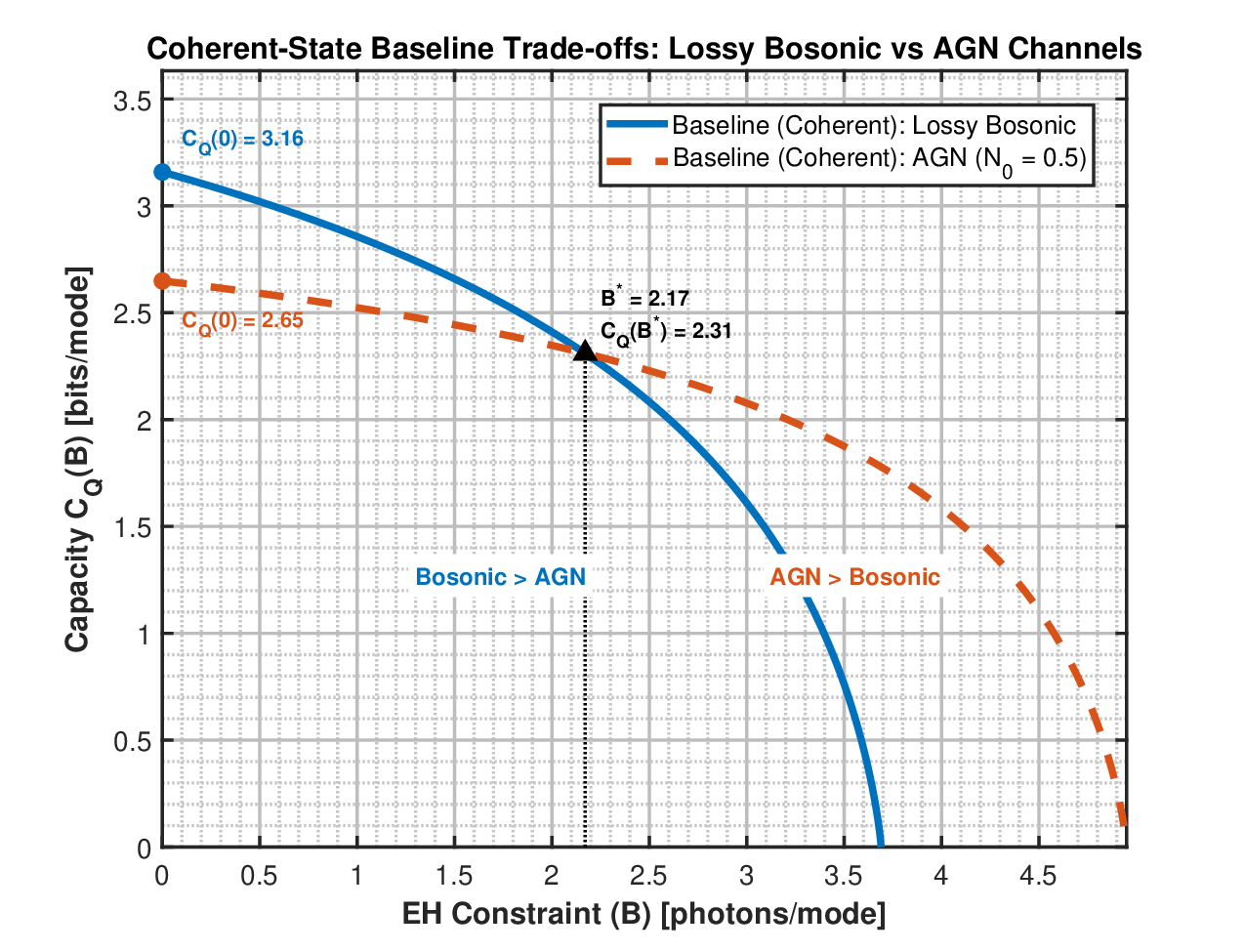}
     \caption{{Coherent-state capacity-power trade-offs of the lossy bosonic ($T=0.8$, $n_{\mathrm{env}}=0.5$) and AGN ($N_0=0.5$) channels under a common transmitter photon-number budget $N_{\max}=5$ and EH efficiency $\eta_h=0.9$.}}
    \label{fig:bosvsAGN}
\end{figure}

\section{Conclusion}

{This paper proposes a unified framework for QSIPT, introducing a quantum-theoretic approach to characterize the fundamental limits of simultaneous information and energy transfer in quantum communication systems. We studied the capacity-power function $C_Q(B)$, which quantifies the fundamental trade-off between classical information rate and harvested energy for both DV and CV quantum channels. For DV channels, analytical upper and lower bounds on $C_Q(B)$ are derived for the amplitude damping channel, while an exact closed-form characterization is obtained for the quantum erasure channel, where $C_Q(B)$ is shown to be strictly concave and monotonically decreasing in the active EH regime, governed by progressive biasing of the optimal input ensemble. For CV channels, a generalized BS receiver with jointly optimized adjustable transmissivity and transmitter photon-number budget yields a well-posed framework. Specifically,  closed-form capacity-power functions are established for both the lossy bosonic and AGN channels under coherent-state encoding. The trade-off is realized entirely through the receiver architecture: the optimal transmissivity $\tau^\star(B)$ decreases affinely with the EH requirement, producing a strictly concave, monotonically decreasing $C_Q(B)$. Within the considered displaced Gaussian encoding class, coherent-state encoding is found to be the best-performing input across all tested parameter configurations, a finding consistent with the phase-insensitive nature of the channel and the passive receiver architecture, which do not enable the exploitation of quadrature asymmetry introduced by squeezing. A comparison of the two CV models reveals a single interior crossing of their capacity-power curves, identifying a rate-limited regime favoring the lossy bosonic channel and an energy-limited regime favoring the AGN channel, a structural property of the distinct roles of multiplicative attenuation and additive noise in QSIPT. These results provide a mathematical framework for the design of quantum communication systems operating under EH constraints.} {The proposed capacity-power function naturally extends to multi-user and network quantum communication settings, where it generalizes to a capacity-power region characterizing the jointly achievable information rates and harvested energy levels across users. In quantum broadcast and multiple-access channels, this region couples message rates and individual or common energy constraints, requiring a balance between signal distinguishability for reliable decoding and coherent signal combination for efficient energy transfer \cite{Winter_2001,Sen_2012}. More general network scenarios, such as quantum interference channels, introduce additional challenges due to quantum interference and coordination, making the complete characterization of multi-user QSIPT an open problem for future research \cite{Carleial_1978,Wilde_2013}.} {From an application perspective, the proposed QSIPT framework provides fundamental performance limits for energy-constrained quantum receivers operating over canonical DV and CV channel models, such as amplitude damping, lossy bosonic, and additive Gaussian noise channels. These models are widely used abstractions of practical optical and solid-state quantum communication links, capturing  physical effects including spontaneous emission loss, optical attenuation, and phase-insensitive detection noise \cite{b1, b18, Pirandola_2018}. In this context, the capacity-power function  characterizes the fundamental trade-off between classical information extraction and usable energy transfer in scenarios such as low-power quantum sensing and short-range optical quantum links, where receivers must operate under strict photon or energy budgets \cite{Degen_2017,Wang_2020}. While the present work focuses on information-theoretic characterization rather than hardware realization, experimental implementation and validation of QSIPT architectures remain important directions for future research.}

\begin{appendices}
\section{Proof of Theorem 1}
\label{appendix:ad_bounds}

{The proof establishes the upper and lower bounds separately, without claim of equality between them.}

\subsubsection{{Upper Bound (Converse)}}
{\textbf{Step~1 (diagonal inputs are without loss of generality for the converse).}
Consider an arbitrary input state
\[
\rho=\begin{pmatrix} a & c\\ c^* & 1-a\end{pmatrix},\qquad |c|^2\le a(1-a).
\]
Under $\mathcal{N}_\gamma$,
\[
\mathcal{N}_\gamma(\rho) =
\begin{pmatrix}
a+\gamma(1-a) & \sqrt{1-\gamma}\,c\\
\sqrt{1-\gamma}\,c^* & (1-\gamma)(1-a)
\end{pmatrix}.
\]
The harvested energy $\mathrm{Tr}[H\mathcal{N}_\gamma(\rho)]=(1-\gamma)(1-a)$ depends only on $a$ and is independent of $c$. {Setting $d_0:=a+\gamma(1-a)$ and $d_1:=(1-\gamma)(1-a)$ for the diagonal entries of $\mathcal{N}_\gamma(\rho)$,} the eigenvalues of $\mathcal{N}_\gamma(\rho)$ are
\[
\lambda_\pm = \tfrac12\!\left(1\pm\sqrt{({d_0}-{d_1})^2+4(1-\gamma)|c|^2}\right),
\]
where ${d_0}-{d_1}=(2a-1)+2\gamma(1-a)$. The larger eigenvalue $\lambda_+\ge\tfrac12$ is strictly increasing in $|c|^2$ (since $1-\gamma>0$), and since $H_2$ is strictly decreasing on $(\tfrac12,1)$, the output entropy $S(\mathcal{N}_\gamma(\rho))=H_2(\lambda_+)$ is strictly decreasing in $|c|^2$. Hence, for any fixed $a$ satisfying the EH constraint, the output entropy is maximized at $c=0$, yielding the diagonal output matrix $\mathrm{diag}({d_0},{d_1})$ with entropy $H_2({d_0})=H_2({d_1})=H_2((1-\gamma)(1-a))$.}

{\textbf{Step~2 (derivation of the upper bound).}
For any ensemble $\{p_x,\rho_x\}$ with average state $\bar{\rho}$, the Holevo quantity satisfies
\begin{equation}
\chi=S(\mathcal{N}_\gamma(\bar{\rho}))-\sum_x p_x S(\mathcal{N}_\gamma(\rho_x))\le S(\mathcal{N}_\gamma(\bar{\rho})).
\end{equation}
Let $\bar{a}$ denote the ground-state population of $\bar{\rho}$. The harvested energy constraint $\mathrm{Tr}[H\mathcal{N}_\gamma(\bar{\rho})]=(1-\gamma)(1-\bar{a})\ge B$ gives $\bar{a}\le 1-B/(1-\gamma)$. By Step~1, $S(\mathcal{N}_\gamma(\bar{\rho}))$ is maximized by setting the off-diagonal of $\bar{\rho}$ to zero, at which point $S(\mathcal{N}_\gamma(\bar{\rho}))=H_2((1-\gamma)(1-\bar{a}))$. Maximizing over $\bar{a}\in[0,1-B/(1-\gamma)]$ establishes the upper bound~\eqref{eq:upper_bound_AD}.}

\subsubsection{{Lower Bound (Achievability via Binary Ensemble)}}
{\textbf{Step~3 (the binary ensemble and its Holevo information).}
Consider the pure-state binary ensemble $\mathcal{E}_{\mathrm{bin}}=\{(\tfrac12,|0\rangle\langle 0|),(\tfrac12,|1\rangle\langle 1|)\}$. The amplitude damping channel maps each component as
\[
\mathcal{N}_\gamma(|0\rangle\langle 0|)=|0\rangle\langle 0|,\qquad
\mathcal{N}_\gamma(|1\rangle\langle 1|)=\mathrm{diag}(\gamma,1-\gamma).
\]
The average output state is
\[
\bar{\rho}_{\mathrm{out}}=\tfrac12|0\rangle\langle 0|+\tfrac12\mathcal{N}_\gamma(|1\rangle\langle 1|)=\mathrm{diag}\!\left(\tfrac{1+\gamma}{2},\tfrac{1-\gamma}{2}\right),
\]
with $S(\bar{\rho}_{\mathrm{out}})=H_2((1-\gamma)/2)$. Since $S(|0\rangle\langle 0|)=0$ and $S(\mathcal{N}_\gamma(|1\rangle\langle 1|))=H_2(1-\gamma)$, the Holevo information is
\begin{equation}
\chi(\mathcal{E}_{\mathrm{bin}})=H_2\!\left(\tfrac{1-\gamma}{2}\right)-\tfrac12 H_2(1-\gamma).
\end{equation}}

{\textbf{Step~4 (verification of the EH constraint).}
The expected output energy of $\mathcal{E}_{\mathrm{bin}}$ is $\mathrm{Tr}[H\bar{\rho}_{\mathrm{out}}]=(1-\gamma)/2$, so $\mathcal{E}_{\mathrm{bin}}$ is feasible for all $B\le(1-\gamma)/2$. Since $C_Q(B)$ is non-increasing in $B$ (Proposition~1), the lower bound~\eqref{eq:lower_bound_AD} holds throughout this range. This completes the proof.}

 \section{Proof of Proposition 5}
 \label{ProofOfProp}
 The classical capacity of the quantum erasure channel is determined by the {Holevo bound}, which provides an upper limit on the amount of classical information that can be transmitted through a quantum channel.

The Holevo bound is given by the following general expression
\begin{equation}
\chi = S\left( \sum_i p_i \rho_i \right) - \sum_i p_i S(\rho_i),
\end{equation}
where  \( \rho_i \) are the input quantum states that are transmitted through the channel,
     \( p_i \) are the probabilities of sending the corresponding states \( \rho_i \),
     \( S(\rho) \) is the Von Neumann entropy of a quantum state \( \rho \), defined as in \eqref{eq:vn_entropy}, and  \( \sum_i p_i S(\rho_i) \) is the average entropy of the individual states. For the quantum erasure channel, qubits are either transmitted successfully with probability \( 1 - p_e \), or they are erased with probability \( p_e \), where \( p_e \) is the erasure probability. The receiver is informed if the qubit is erased through an erasure flag. We assume the input states \( \rho_i \) are pure states. In the case of qubits, we can assume two typical input states \( |0\rangle \) and \( |1\rangle \), represented as density matrices
\begin{equation}
\rho_0 = |0\rangle \langle 0|, \quad \rho_1 = |1\rangle \langle 1|.
\end{equation}
Since these are pure states, their Von Neumann entropy is given by
\begin{equation}
S(\rho_i) = 0 \quad \text{(for pure states)}, \hspace{1ex} i\in\{0,1\}.
\end{equation}

The quantum erasure channel transforms the input state \( \rho \) into an output state \( \mathcal{N}(\rho) \) according to the following rule
\begin{equation}
\mathcal{E}(\rho) = (1 - p_e) \rho + p_e |e\rangle \langle e|,
\end{equation}
where  \( p_e \) is the erasure probability and 
     \( |e\rangle \) is the erasure flag state that indicates the qubit has been erased. Specifically,  the state is transmitted successfully with probability \( 1 - p_e \), while with probability \( p_e \), it is replaced by the erasure flag.
The ensemble-averaged state is given by the mixture of all possible output states after transmission through the channel. For an ensemble of pure input states \( \rho_i \) with probabilities \( p_i \), the output ensemble is written as 
\begin{equation}
\rho_{\text{avg}} = (1 - p_e) \sum_i p_i \rho_i + p_e |e\rangle \langle e|.
\end{equation}
Since \( S(\rho_i) = 0 \) for the  pure states, we only need to compute the entropy of the ensemble-averaged output state.
The output state of the channel is a probabilistic mixture of the transmitted quantum state and the erasure flag state. The entropy of the output state \( \mathcal{N}(\rho) \) is computed as
\begin{equation}
S(\mathcal{N}(\rho)) = H(p_e),
\end{equation}
where \( H(p_e) \) is the Shannon entropy of the probability distribution \( \{ p_e, 1 - p_e \} \), given by
\begin{equation}
H(p_e) = -p_e \log_2 p_e - (1 - p_e) \log_2 (1 - p_e).
\end{equation}
Since the input states \( \rho_i \) are pure states, their Von Neumann entropy \( S(\rho_i) \) is zero. Therefore, the Holevo bound simplifies to
\begin{equation}
\chi = S(\rho) = H(p_e).
\end{equation}

The classical capacity \( C \) of the quantum erasure channel is obtained by maximizing the Holevo bound over all possible input ensembles. For qubits, the maximum capacity is achieved when the input states are maximally mixed, and the capacity is given by
\begin{equation}
C = (1 - p_e) \log_2 d,
\end{equation}
where \( d \) is the dimension of the quantum system. For a single qubit quantum system with \( d = 2 \), and the capacity simplifies to
\begin{equation}
C = 1 - p_e.
\end{equation}

\section{Proof of Theorem 2}
\label{appendix:erasure}
{We adopt the standard qubit normalization $E_0=0$, $E_1=1$ throughout, and let $E_e\ge 0$ denote the erasure-flag energy.} We parametrize the input state as follows
\begin{equation}
\rho = \begin{pmatrix}
a & c \\
c^* & 1 - a
\end{pmatrix}, \quad a \in [0, 1], \quad |c|^2 \leq a(1 - a), \label{eq:rho_erasure}
\end{equation}
where \( a \) is the ground-state population. The output of the erasure channel is
        \begin{equation}
\rho_{\mathrm{out}} = (1 - p_e) \rho + p_e |e\rangle \langle e|. \label{eq:erasure_output}
\end{equation}

{\textbf{Step~1 (reduction to diagonal inputs).}
Since $|e\rangle$ is orthogonal to the qubit subspace $\mathrm{span}\{|0\rangle,|1\rangle\}$, the Holevo information of the output ensemble satisfies
\begin{equation}
\chi=(1-p_e)\,\chi_{\mathrm{in}},
\end{equation}
where $\chi_{\mathrm{in}}=S(\rho)-\sum_x p_x S(\rho_x)$ denotes the Holevo information of the input ensemble. For the optimization of $\chi_{\mathrm{in}}$, it suffices to restrict the optimization to pure-state ensembles. The von Neumann entropy of $\rho$ is $S(\rho)=H_2(\lambda)$ with
\[
\lambda=\tfrac12\bigl(1+\sqrt{(2a-1)^2+4|c|^2}\bigr),
\]
which is maximized for fixed $a$ at $c=0$, yielding $S(\rho)\big|_{c=0}=H_2(a)$. Therefore,
\begin{equation}
C_Q(B)=\!\!\!\max_{\substack{0\le a\le 1\\ (1-p_e)(1-a)+p_eE_e\ge B}}\!\!\!(1-p_e)H_2(a).
\label{eq:erasure_cq}
\end{equation}}

{\textbf{Step~2 (EH constraint and $a_{\max}(B)$).}
With $E_0=0$, $E_1=1$, the expected output energy is
\begin{IEEEeqnarray}{rCl}
\nonumber
\mathrm{Tr}[H\rho_{\mathrm{out}}]&=&(1-p_e)\bigl(a\cdot 0+(1-a)\cdot 1\bigr)+p_eE_e\\
\label{eq:erasure_energy}
&=&(1-p_e)(1-a)+p_eE_e.
\end{IEEEeqnarray}
The constraint $\mathrm{Tr}[H\rho_{\mathrm{out}}]\ge B$ is equivalent to
\begin{equation}
a\le a_{\max}(B):=1-\frac{B-p_eE_e}{1-p_e}.
\label{eq:erasure_amax}
\end{equation}
Note that $a_{\max}(B)$ is affine and strictly decreasing in $B$, and feasibility requires $B\le(1-p_e)+p_eE_e$.}

{\textbf{Step~3 (optimization over $a$).}
The optimization in \eqref{eq:erasure_cq} reduces to $\max_{0\le a\le a_{\max}(B)}(1-p_e)H_2(a)$. Since $H_2$ is symmetric about $a=\tfrac12$ with unique maximum $H_2(\tfrac12)=1$:}

{\emph{Case~1 (inactive EH constraint).}
If $a_{\max}(B)\ge\tfrac12$, the unconstrained maximizer $a^\star=\tfrac12$ is feasible. The condition is
\begin{equation}
1-\frac{B-p_eE_e}{1-p_e}\ge\tfrac12\;\Longleftrightarrow\; B\le \frac{1-p_e}{2}+p_eE_e.
\end{equation}
In this regime, $C_Q(B)=(1-p_e)H_2(\tfrac12)=1-p_e$.}

{\emph{Case~2 (active EH constraint).}
If $B>(1-p_e)/2+p_eE_e$, then $a_{\max}(B)<\tfrac12$, and the constrained maximum is at the boundary $a=a_{\max}(B)$. Using the symmetry $H_2(x)=H_2(1-x)$:
\begin{equation}
C_Q(B)=(1-p_e)H_2(a_{\max}(B))=(1-p_e)H_2\!\left(\frac{B-p_eE_e}{1-p_e}\right).
\end{equation}
Combining both cases yields the formula in Theorem~2.}

\section{Proof of Theorem \ref{thm:bosonic_coh}}
\label{ProofofHelovo}
{Consider a lossy bosonic channel with transmissivity $T\in[0,1]$ and environmental thermal noise with mean photon number $n_{\mathrm{env}}$, followed by a generalized BS receiver with adjustable transmissivity $\tau\in[0,1]$. For a coherent-state ensemble with average input photon number $N_s\in[0,N_{\max}]$, the mean photon number at the channel output is
\begin{equation}
N_{\mathrm{out}}^{\mathrm{bos}}(N_s)=TN_s+(1-T)n_{\mathrm{env}},
\label{eq:appE_Nout}
\end{equation}
each output state is a displaced thermal state with identical covariance, and the comm-port Holevo information is given by~\eqref{eq:bosonic_holevo_comm}. The harvested energy per channel use, modeled at the EH port of the generalized BS, is
\begin{equation}
E_h(N_s,\tau)=\eta_h(1-\tau) {N_{\mathrm{out}}^{\mathrm{bos}}(N_s)}.
\label{eq:appE_EH}
\end{equation}
We establish Theorem~\ref{thm:bosonic_coh} through two lemmas, corresponding to the optimal photon number and the optimal BS transmissivity, respectively.}

{\textbf{Lemma~D.1 (Optimality of $N_s^\star=N_{\max}$).}
\emph{For fixed $\tau\in(0,1]$, the comm-port Holevo information ${\chi_{\mathrm{comm}}^{\mathrm{coh,bos}}}(N_s,\tau)$ is strictly increasing in $N_s\in[0,N_{\max}]$. Consequently, the optimal photon number is $N_s^\star=N_{\max}$.}}

{\emph{Proof.} For fixed $\tau\in(0,1]$,
\[
{\chi_{\mathrm{comm}}^{\mathrm{coh,bos}}}(N_s,\tau)=g\bigl(\tau {N_{\mathrm{out}}^{\mathrm{bos}}}(N_s)\bigr)-g\bigl(\tau(1-T)n_{\mathrm{env}}\bigr),
\]
where ${N_{\mathrm{out}}^{\mathrm{bos}}}(N_s)=TN_s+(1-T)n_{\mathrm{env}}$ is strictly increasing and affine in $N_s$ (since $T>0$). The bosonic entropy $g(x)=(x+1)\log_2(x+1)-x\log_2 x$ is strictly increasing on $[0,\infty)$. By composition, $g(\tau {N_{\mathrm{out}}^{\mathrm{bos}}}(N_s))$ is strictly increasing in $N_s$, while $g(\tau(1-T)n_{\mathrm{env}})$ is independent of $N_s$. Hence ${\chi_{\mathrm{comm}}^{\mathrm{coh,bos}}}(N_s,\tau)$ is strictly increasing in $N_s$. Moreover, for fixed $\tau$, the harvested energy $E_h(N_s,\tau)$ in~\eqref{eq:appE_EH} is also strictly 
increasing in $N_s$; hence every feasible point remains feasible as $N_s$ 
increases, and the feasible set $\{N_s\in[0,N_{\max}]:E_h(N_s,\tau)\ge B\}$ is  non-empty for all $B\in[0,B_{\max}^{\mathrm{bos}}]$. Since the objective is  strictly increasing given the budget constraint $N_s\le N_{\max}$, the  constrained maximum is attained uniquely at $N_s^\star=N_{\max}$. 
\hfill$\blacksquare$}

{\textbf{Lemma~D.2 (Strict monotonicity in $\tau$).}
\emph{With $N_s=N_{\max}$, the Holevo function $f(\tau):=g(\tau N_{\mathrm{out}}^{\max,\mathrm{bos}})-g(\tau(1-T)n_{\mathrm{env}})$ is strictly increasing in $\tau\in(0,1]$.}}

{\emph{Proof.} Differentiating,
\[
f'(\tau)=N_{\mathrm{out}}^{\max,\mathrm{bos}}\,g'(\tau N_{\mathrm{out}}^{\max,\mathrm{bos}})-(1-T)n_{\mathrm{env}}\,g'(\tau(1-T)n_{\mathrm{env}}),
\]
where $g'(x)=\log_2(1+1/x)$ for $x>0$. Define $\varphi(x):=xg'(x)=x\log_2(1+1/x)$. Then $f'(\tau)=\frac{1}{\tau}\!\left[\varphi(\tau N_{\rm out}^{\max,\rm bos})-\varphi(\tau(1-T)n_{\mathrm{env}})\right].$ Since
\[
\varphi'(x)=\log_2(1+1/x)-\frac{1}{(x+1)\ln 2}>0,\quad x>0
\]
(which follows from the standard inequality $\ln(1+t)>t/(1+t)$ for $t>0$), $\varphi$ is strictly increasing. Since $N_{\mathrm{out}}^{\max,\mathrm{bos}}=TN_{\max}+(1-T)n_{\mathrm{env}}>(1-T)n_{\mathrm{env}}$ for $T>0$ and $N_{\max}>0$, we have $\tau N_{\mathrm{out}}^{\max,\mathrm{bos}}>\tau(1-T)n_{\mathrm{env}}$, and by the strict monotonicity of $\varphi$, $f'(\tau)>0$ for all $\tau>0$. Therefore $f$ is strictly increasing in $\tau$. \hfill$\blacksquare$}

{\textbf{Conclusion of the proof.}
With $N_s^\star=N_{\max}$, the EH constraint reduces to
\[
\eta_h(1-\tau)N_{\mathrm{out}}^{\max,\mathrm{bos}}\ge B\;\Longleftrightarrow\;\tau\le 1-\frac{B}{\eta_h N_{\mathrm{out}}^{\max,\mathrm{bos}}}=:\tau_{\max}(B).
\]
By Lemma~D.2, $f$ is strictly increasing in $\tau$, so the constrained maximum is attained at the feasibility boundary $\tau^\star(B)=\tau_{\max}(B)=1-B/(\eta_h N_{\mathrm{out}}^{\max,\mathrm{bos}})$. {Since ${C_Q^{\mathrm{coh,bos}}}(B)=f(\tau^\star(B))$ by definition of $f(\tau):=g(\tau N_{\mathrm{out}}^{\max,\mathrm{bos}})-g(\tau(1-T)n_{\mathrm{env}})$,} this yields
\begin{equation}
{C_Q^{\mathrm{coh,bos}}}(B)=g(\tau^\star(B)\,N_{\mathrm{out}}^{\max,\mathrm{bos}})-g(\tau^\star(B)(1-T)n_{\mathrm{env}}).
\end{equation}
Strict decrease and strict concavity of ${C_Q^{\mathrm{coh,bos}}}(B)$ in $B$ follow by composition. The map $\tau^\star(B)=1-B/(\eta_h N_{\mathrm{out}}^{\max,\mathrm{bos}})$ is affine and strictly decreasing in $B$. Using $g''(x)=-1/[x(x+1)\ln 2]<0$ for all $x>0$, a direct computation gives
\[
f''(\tau)=\frac{(1-T)n_{\mathrm{env}}-N_{\mathrm{out}}^{\max,\mathrm{bos}}}{\tau\,\ln 2\,(\tau N_{\mathrm{out}}^{\max,\mathrm{bos}}+1)(\tau(1-T)n_{\mathrm{env}}+1)},
\]
which is strictly negative since $N_{\mathrm{out}}^{\max,\mathrm{bos}}>(1-T)n_{\mathrm{env}}$ (valid for $T>0$, $N_{\max}>0$). Hence $f$ is strictly concave on $(0,1]$, and by composition with the affine $\tau^\star(B)$, ${C_Q^{\mathrm{coh,bos}}}(B)$ is strictly concave and strictly decreasing in $B$. The boundary values ${C_Q^{\mathrm{coh,bos}}}(0)=g(N_{\mathrm{out}}^{\max,\mathrm{bos}})-g((1-T)n_{\mathrm{env}})$ at $\tau^\star=1$ and ${C_Q^{\mathrm{coh,bos}}}(\eta_h N_{\mathrm{out}}^{\max,\mathrm{bos}})=0$ at $\tau^\star=0$ complete the proof.}
\section{Proof of Proposition \ref{prop:bosonic_sq}}
\label{ProofOfBosSq}
{We derive the comm-port Holevo information of the displaced squeezed-state ensemble at the BS comm port, accounting for both the channel attenuation and the BS vacuum noise. Throughout the derivation, the input photon-number budget $N_s\le N_{\max}$ and the BS transmissivity $\tau\in[0,1]$ are jointly optimized in accordance with~\eqref{eq:cv_capacity_power}.}

{\textbf{Channel-output covariance (before BS).}
For a displaced squeezed state with squeezing parameter $r\ge 0$ and modulation variances $\sigma_X^2,\sigma_P^2\ge 0$, the conditional output state has quadrature variances
\begin{align}
\mathrm{Var}(\hat{X}_{\mathrm{out}}^{\mathrm{single}}) &=\tfrac12\bigl[Te^{-2r}+(1-T)(2n_{\mathrm{env}}+1)\bigr],\\
\mathrm{Var}(\hat{P}_{\mathrm{out}}^{\mathrm{single}}) &=\tfrac12\bigl[Te^{2r}+(1-T)(2n_{\mathrm{env}}+1)\bigr].
\end{align}
The displacement modulation $\alpha\sim\mathcal{CN}(0,\Sigma)$ with $\Sigma=\mathrm{diag}(\sigma_X^2,\sigma_P^2)$ is attenuated by $\sqrt{T}$ through the channel, so the average output covariance before the BS reads
\begin{align}
\mathrm{Var}(\hat{X}_{\mathrm{out}}^{\mathrm{avg}}) &=\mathrm{Var}(\hat{X}_{\mathrm{out}}^{\mathrm{single}})+T\sigma_X^2,\\
\mathrm{Var}(\hat{P}_{\mathrm{out}}^{\mathrm{avg}}) &=\mathrm{Var}(\hat{P}_{\mathrm{out}}^{\mathrm{single}})+T\sigma_P^2.
\end{align}}

{\textbf{Comm-port covariance (after BS).}
The generalized BS with transmissivity $\tau\in[0,1]$ adds vacuum noise $(1-\tau)/2$ per quadrature. The comm-port quadrature variances are accordingly
\begin{align}
V_X^{\mathrm{single}} &= \tau\,\mathrm{Var}(\hat{X}_{\mathrm{out}}^{\mathrm{single}})+\tfrac{1-\tau}{2},\\
V_P^{\mathrm{single}} &= \tau\,\mathrm{Var}(\hat{P}_{\mathrm{out}}^{\mathrm{single}})+\tfrac{1-\tau}{2},\\
V_X^{\mathrm{avg}} &= V_X^{\mathrm{single}}+\tau T\sigma_X^2,\\
V_P^{\mathrm{avg}} &= V_P^{\mathrm{single}}+\tau T\sigma_P^2.
\end{align}
The factor $\tau T$ in the last two lines is critical and arises because the displacement is attenuated by $\sqrt{T}$ through the channel and then by $\sqrt{\tau}$ at the BS.}

{\textbf{Comm-port Holevo information.}
The symplectic eigenvalues are $\nu_{\mathrm{single}}=\sqrt{V_X^{\mathrm{single}}V_P^{\mathrm{single}}}$ and $\nu_{\mathrm{avg}}=\sqrt{V_X^{\mathrm{avg}}V_P^{\mathrm{avg}}}$. The comm-port Holevo information is
\begin{equation}
{\chi_{\mathrm{comm}}^{\mathrm{sq,bos}}}(r,\sigma_X^2,\sigma_P^2,\tau)=g(\nu_{\mathrm{avg}}-\tfrac12)-g(\nu_{\mathrm{single}}-\tfrac12).
\label{eq:appF_chi_sq}
\end{equation}
The total mean input photon number is $N_s=\sinh^2(r)+(\sigma_X^2+\sigma_P^2)/2$, and the EH constraint at the EH port reads
\begin{equation}
\eta_h(1-\tau){N_{\mathrm{out}}^{\mathrm{bos}}}(N_s)\ge B.
\label{eq:appF_EH}
\end{equation}}

{
Maximizing~\eqref{eq:appF_chi_sq} over $(r,\sigma_X^2,\sigma_P^2,\tau)$ subject to $N_s\le N_{\max}$ and~\eqref{eq:appF_EH} yields the achievable-rate-power function in~\eqref{eq:cv_sq_optim}. Since this objective does not admit a closed form for $r>0$, the optimization is performed numerically.}

\section{Proof of Proposition 7}
\label{ProofOfAGNSQ}
{We adapt the derivation of Appendix~\ref{ProofOfBosSq} to the AGN channel, which is the limiting case of the lossy bosonic channel with $T=1$ and additive noise variance $N_0$. The output mode $\hat{a}_{\mathrm{out}}$ enters the generalized BS receiver of Section~IV-A with adjustable transmissivity $\tau\in[0,1]$.} We consider a squeezed vacuum state as described in \eqref{eqSqrho} drawn from a zero mean Gaussian distribution with quadrature variances $\sigma_X^2$ and $\sigma_P^2$. The total mean photon number of the input ensemble is given by
\begin{equation}
N_s = \sinh^2(r) + \frac{1}{2} (\sigma_X^2 + \sigma_P^2).
\end{equation}
The AGN channel adds phase-insensitive Gaussian noise of variance $N_0$ to both quadratures, so that the conditional output quadrature variances (before the BS) are
\begin{align}
\mathrm{Var}(\hat{X}_{\mathrm{out}}^{\mathrm{single}}) &= \tfrac{1}{2} e^{-2r} + N_0, \\
\mathrm{Var}(\hat{P}_{\mathrm{out}}^{\mathrm{single}}) &= \tfrac{1}{2} e^{2r} + N_0.
\end{align}
The displacement modulation $\alpha\sim\mathcal{CN}(0,\mathrm{diag}(\sigma_X^2,\sigma_P^2))$ is not attenuated in the AGN channel, so the average output covariance before the BS is obtained by adding $\sigma_X^2$ and $\sigma_P^2$ to the corresponding conditional variances. {The generalized BS with transmissivity $\tau$ adds vacuum noise $(1-\tau)/2$ per quadrature, yielding the comm-port quadrature variances
\begin{align}
V_X^{\mathrm{single}} &= \tau\,\mathrm{Var}(\hat{X}_{\mathrm{out}}^{\mathrm{single}})+\tfrac{1-\tau}{2},\\
V_P^{\mathrm{single}} &= \tau\,\mathrm{Var}(\hat{P}_{\mathrm{out}}^{\mathrm{single}})+\tfrac{1-\tau}{2},\\
V_X^{\mathrm{avg}} &= V_X^{\mathrm{single}}+\tau\sigma_X^2,\\
V_P^{\mathrm{avg}} &= V_P^{\mathrm{single}}+\tau\sigma_P^2,
\end{align}
where the multiplicative factor $\tau$ on the modulation variances reflects the BS attenuation of the displacement (no channel attenuation, since $T=1$ in the AGN model).} The conditional output state is a mixed Gaussian state with comm-port symplectic eigenvalue
\begin{equation}
\nu_{\text{single}} = \sqrt{V_X^{\mathrm{single}}\cdot V_P^{\mathrm{single}}}.
\end{equation}
The average comm-port covariance has symplectic eigenvalue
\begin{equation}
\nu_{\text{avg}} = \sqrt{V_X^{\mathrm{avg}}\cdot V_P^{\mathrm{avg}}}.
\end{equation}
{The comm-port Holevo information of the displaced squeezed-state ensemble at the BS comm port is
\begin{equation}
\chi_{\mathrm{comm}}^{\mathrm{sq,AGN}}(r,\sigma_X^2,\sigma_P^2,\tau)=g(\nu_{\text{avg}}-\tfrac12)-g(\nu_{\text{single}}-\tfrac12).
\end{equation}
The EH constraint at the EH port reads
\begin{equation}
\eta_h(1-\tau)(N_s+N_0)\ge B,
\label{eq:appG_EH}
\end{equation}
and the achievable-rate-power function is obtained by maximizing $\chi_{\mathrm{comm}}^{\mathrm{sq,AGN}}$ over $(r,\sigma_X^2,\sigma_P^2,\tau)$ subject to $N_s\le N_{\max}$ and~\eqref{eq:appG_EH}. As in the lossy bosonic case (Appendix~\ref{ProofOfBosSq}), the optimization does not admit a closed form for $r>0$ and is performed numerically; numerical results consistently yield $r^\star=0$ within the displaced Gaussian encoding class.}
 \end{appendices}

\end{document}